\pdfoutput=1
\RequirePackage{boolean}
\RequirePackage{reversion}

\newcommand{\Anonymous}{\False}
\newcommand{\cameraversion}{\False}
\newcommand{\Draft}{\False}

\newcommand{\Beamer}{\False}
\newcommand{\Print}{\False}
\newcommand{\Long}{\True}

\cameraversion{
  \Draft{\newcommand{\sigplanoptions}{preprint}}
        {\newcommand{\sigplanoptions}{}}
}{
  \Draft{\newcommand{\sigplanoptions}{nocopyrightspace,preprint}}
        {\newcommand{\sigplanoptions}{nocopyrightspace}}
}

\documentclass[\sigplanoptions]{sigplanconf}

\usepackage{amsmath}
\usepackage{amsthm}
\usepackage{lib}

\newtheorem{theorem}{Theorem}
\newtheorem{lemma}{Lemma}
\newtheorem{corollary}{Corollary}
\newtheorem{definition}{Definition}
\newtheorem*{definition*}{Definition}

\begin{document}
\cameraversion{\toappear{}}{}

\setlength{\pdfpageheight}{\paperheight}
\setlength{\pdfpagewidth}{\paperwidth}

\conferenceinfo{CONF 'yy}{Month d--d, 20yy, City, ST, Country}
\copyrightyear{20yy}
\copyrightdata{978-1-nnnn-nnnn-n/yy/mm}
\copyrightdoi{nnnnnnn.nnnnnnn}

% Uncomment the publication rights you want to use.
%\publicationrights{transferred}
%\publicationrights{licensed}     % this is the default
%\publicationrights{author-pays}

\titlebanner{Draft of \today}        % These are ignored unless
\preprintfooter{}   % 'preprint' option specified.

\title{Deciding equivalence with sums and the empty type}
\subtitle{}

\Anonymous{
  \authorinfo{Anonymous for review}
             {}
             {}
}{
  \authorinfo{Gabriel Scherer}
             {Northeastern University, USA}
             {gabriel.scherer@gmail.com}
}
% \authorinfo{Name2\and Name3}
%            {Affiliation2/3}
%            {Email2/3}

\maketitle

\begin{abstract}
  The logical technique of \emph{focusing} can be applied to the
  $\lambda$-calculus; in a simple type system with atomic types and
  negative type formers (functions, products, the unit type), its
  normal forms coincide with $\beta\eta$-normal forms. Introducing
  a \emph{saturation} phase gives a notion of quasi-normal forms in
  presence of positive types (sum types and the empty type). This rich
  structure let us prove the decidability of $\beta\eta$-equivalence
  in presence of the empty type, the fact that it coincides with
  contextual equivalence, and a finite model property.
\end{abstract}

\begin{version}{\Not\Long}
\mprset{lineskip=0em}
\noabovedisplayskip
\nobelowdisplayskip
\end{version}

\begin{version}{\False}
% I'm not cutting more content of my work
% to show a full paragraph of ACM categories and keywords

\end{version}

\section{Introduction}

\subsection{Notion of equivalences}

For a given type system, there may be several notions of program
equivalence of interest.
One may define a notion of syntactic equivalence by a system of
equations between terms, such as $\beta$-equivalence and
$\eta$-equivalence, or their union $\betaeta$-equivalence.
A more extensional notion of equivalence is contextual or
observational equivalence, which checks that the two terms behave in
the same way under all contexts.
Finally, a semantic notion of equivalence is given by interpreting
terms in a certain mathematical space, typically morphisms in
categories with a certain structure, and considering two terms
equivalent if they are equal under all interpretations.
One can generally prove that considering all interpretations in the
category of sets suffices to distinguish two observably distinct
terms, and certain type systems have the stronger \emph{finite model}
property that considering all interpretations in \emph{finite} sets
suffices to distinguish observably distinct terms.

Contextual equivalence has a clear, compact definition, it is the
notion of equivalence that corresponds to programmers' intuition, but
it is difficult to prove and reason about.
Syntactic equivalence can be used as part of syntactic typing
judgments, and may be easier to prove decidable.
Semantic models provide a more abstract point of view on the identity
of programs, and may enable the powerful method of
normalization-by-evaluation.

For the untyped $\lambda$-calculus, \citet*{boehm1968} proved that
$\betaeta$-equivalence and observational equivalence coincide.
Untyped $\beta$-reduction is not normalizing, so equivalence is undecidable.

For the simply-typed $\lambda$-calculus with atomic types, functions
and pairs -- what we call the \emph{negative} fragment -- we also know
that these notions of equivalence coincide.
Furthermore, typed equivalence is decidable: we can define and compute
$\beta$-short $\eta$-long normal forms, and two terms are
$\betaeta$-equivalent exactly when they have the same normal form.
\citet*{friedman1975} proved that two terms equal in all set-theoretic
models are equivalent, and \citet*{statman1982} sharpened the result by
proving that finite sets suffice to distinguish inequivalent terms -- the
finite model property.

This pleasant setting is quickly eroded by moving to richer
programming languages and type systems.
Adding notions of side-effects, for example, makes the general
$\beta\eta$-equivalence unsound.
Even in the realm of pure, total programming languages, adding
parametric polymorphism (System F and beyond) makes
$\beta\eta$-equivalence strictly weaker than contextual equivalence.

\subsection{Sums and empty type}

An interesting middle-ground is the \emph{full} simply-typed
$\lambda$-calculus, with not only atomic types, functions, pairs and
the unit type $\tunit$ but also sum types and the empty type
$\tempty$.
There, results on program equivalence have been surprisingly long to
come, because of deep difficulties caused by mixing functions
(negative connectives) and sums (positive connectives), and the strong
behavior of the empty type: in an inconsistent context, all terms are
equal.

The first decidability result for the system with non-empty sums was
\citet*{ghani1995}, using advanced rewriting techniques later
simplified by \citet*{lindley2007}.
It was followed by normalization-by-evaluation
results~\citep*{altenkirch2001,balat2004} that also established
decidability of equivalence in the non-empty case using the
categorical structure introduced in \citet*{fiore1999}.
Decidability in the presence of the empty type is still open -- we
propose a proof.

Finally, the only known completeness result for set-theoretic models
is \citet*{dougherty2000}; it only holds for non-empty sums, and
relies in an essential way on infinite sets.
The finite model property is only conjectured -- we propose a proof,
including in presence of the empty type.

\subsection{Focusing}

Focusing~\citep*{andreoli1992} is a general technique that uses the
notion of invertibility of inference rules to define a \emph{focused} subset of any
logic that is complete but removes some redundant proofs.

A recent branch of work on \emph{maximal
  multi-focusing}~\citep*{chaudhuri-miller-saurin-2008,chaudhuri2012}
has demonstrated that focused proofs can be further restricted to
become even more canonical: in each application to a specific logic,
the resulting representations are equivalent to existing representations
capturing the identity of proofs -- proof nets for linear logic,
expansion proofs for classical logic.

\citet*{scherer2015} applied focusing to the $\lambda$-calculus with
non-empty sums. Completeness of focusing there means that any term has
a $\betaeta$-equivalent focused term. Their counterpart of maximal
multi-focusing is \emph{saturation}: saturated terms of the focused
$\lambda$-calculus introduce and deconstruct neutral terms of sum type
as early as possible -- when they are \emph{new}, they were not
derivable in the previous saturation phase. Canonicity of this
representation gives yet another decision procedure for
$\betaeta$-equivalence of $\lambda$-terms with non-empty sums.

The present work extends saturated focusing to the full simply-typed
lambda calculus, with units and in particular the empty type. The
suitably extended notion of saturated form retains its canonicity
properties in the full system. This means that $\beta\eta$-equivalence
with empty types is decidable -- by converting terms to their
saturated focused form. From two distinct saturated forms, one can
furthermore build a well-typed context that distinguishes them. This
proves that contextual equivalence implies equality of normal forms,
and thus $\betaeta$-equivalence: any term is $\betaeta$-equivalent to
its normal form, so two terms with the same normal form are
$\betaeta$-equivalent. Our distinguishing context needs only
instantiate atomic types with finite sets, giving a finite model
property.

Extending saturated focusing to the empty type requires adding
a requirement that saturation phases be complete for provability: not
only must they introduce all new neutrals of positive type for use in
the rest of the term, but they must at least introduce one such
neutral for each type deducible in the current context. As
a consequence, we can prove a Saturation Consistency theorem which is
key to our canonicity: if two saturated terms are distinct, then their
context must be consistent.

\subsection{Contributions}

We establish the following results in the simply-typed
$\lambda$-calculus with atoms, functions, pairs, the unit type, sums
and the empty type $\LC{(\sysfull)}$:

\begin{itemize}
\item Saturated terms provide a notion of \emph{quasi}-normal form;
  equivalent quasi-normal forms are not necessarily
  $\alpha$-equivalent, but are related by a local, easily decidable
  relation of \emph{invertible commutation conversions}
  $\parel{\argeq\icc}$.
\item $\betaeta$-equivalence is decidable.
\item $\betaeta$-equivalence and contextual equivalence coincide,
  along with set-theoretic equivalence in all models where atomic
  types are interpreted by closed types.
\item The finite model property holds -- as closed types in this
  system have finitely many inhabitants.
\end{itemize}

\begin{version}{\Not\cameraversion}
  Our notion of $\betaeta$-equivalence is the \emph{strong}
  equivalence on sums. It corresponds to equality of morphisms in the
  free bi-cartesian closed category -- see \citet*{scherer-phd},
  Section 3.2.2.
\end{version}

\subsection{Plan}

\fullref{sec:full-calculus} introduces the full $\lambda$-calculus we
shall consider in this work, along with the various notions of
equivalence ($\betaeta$, contextual, set-theoretic) we will
discuss. We prove some elementary results: $\betaeta$-equivalence
implies contextual equivalence in all closed models, which coincides
with set-theoretic equivalence in all closed models.

\fullref{sec:focusing} presents focusing, before detailing the
focused $\lambda$-calculus extended with the unit type and empty
type. We formulate the computational counterpart of the completeness
theorem: any $\lambda$-term has a $\beta\eta$-equivalent focused term.

\fullref{sec:saturated-calculus} presents the saturated subset of the
focused $\lambda$-calculus as defined by a different system of
inference rules. Compared to the simpler system of
\citet*{scherer2015}, this saturated system is parametrized over
a \emph{selection function} that selects the neutrals to split over
during saturation. The saturated system is, again, computationally
complete with respect to focused terms.

\fullref{sec:saturation-consistency} establishes the main
meta-theoretic property of saturated terms in presence of the empty
type, namely that saturating an inconsistent context will always find
a proof of $\tempty$. In other words, if two saturated terms differ
after a saturation phase instead of both ending on an absurdity
elimination, we know that they differ in consistent contexts. This
result is key to extending the distinguishability result of
\citet*{dougherty2000} to a system with empty types.

Finally, \fullref{sec:canonicity} establishes the central result of
this work: if two saturated $\lambda$-terms $\ea \argneq\icc \eb$ are
syntactically distinct (modulo invertible commuting conversions), then
there exists a closed type model $\ma$ in which we have
a distinguishing context $\plug \Ca \ehole$ such
$\plug \Ca {\modapp \ma \ea}, \plug \Ca {\modapp \ma \eb}$ are
closed booleans ($\sum \tunit \tunit$), one of them equal to $\etrue$
and the other to $\efalse$. By contraposition, this proves that two
terms contextually equivalent in all models have the same saturated
normal forms -- giving decidability -- and in particular are
$\betaeta$-equivalent.

\begin{version}{\Not\cameraversion}
For lack of space, our statements do not come with their proofs, but
proof outlines are given in \fullref{ann:proofs}.

\begin{quotation}
  \textit{Sections 1 to 5, in particular the presentations of focusing
    and saturated normal forms, correspond to content that was
    presented in the thesis manuscript \citet*{scherer-phd}. In the
    present article, definitions and proofs in these sections are
    given with minimal amount of details for economy of space; they
    are presented in full in the manuscript. In contrast, Section 6,
    which presents the key canonicity result enabling the
    contributions of this article, is entirely new, and presented in
    more detail.}
\end{quotation}
\end{version}

\begin{version}{\cameraversion}
  \begin{quotation}
    \textit{See \url{https://arxiv.org/abs/1610.01213} for a less
      space-constained version of this work, with proof arguments in
      appendices, and the manuscript \citet*{scherer-phd} for
      a very detailed exposition of the first five sections.}
  \end{quotation}
\end{version}

\section{Equivalences in the full $\lambda$-calculus}
\label{sec:full-calculus}

\subsection{Typing rules and $\betaeta$-equivalence}

\begin{smallmathparfig}
  {fig:LCfull}
  {Full simply-typed lambda-calculus \small{$\LCfull$}}
  \ta, \tb, \tc \gramdef \xa, \xb, \xc
    \mid \fun \ta \tb
    \mid \prodfam \ta \mid \tunit
    \mid \sumfam \ta \mid \tempty

  \ea, \eb, \ec \gramdef
    \begin{array}{l@{~}l}
      \eva, \evb, \evc &
      \mid \lam \eva \ea \mid \app \ea \eb
      \mid \pairfam \ea \mid \proj \ii \ea \mid \eunit \\
      &
      \mid \inj \ii \ea \mid \fdesumfam \ea \eva \eb \mid \eabsurd \ea
    \end{array}
  \\
  \infer%[stlc-lam]
  {\dernat {\ca, \var \ev \ta} \ea \tb}
  {\dernat \ca {\lam \ev \ea} {\fun \ta \tb}}

  \infer%[stlc-app]
  {\dernat \ca \ea {\fun \ta \tb}\\ \dernat \ca \eb \ta}
  {\dernat \ca {\app \ea \eb} \tb}
  \\
  \infer%[stlc-pair]
  {\dernat \ca \ea \ta\\ \dernat \ca {\ob \eb} \tb}
  {\dernat \ca {\ob {\pair \ea \eb}} {\prod \ta \tb}}

  \infer%[stlc-proj]
  {\dernat \ca \ea {\prodfam \ta}}
  {\dernat \ca {\ob {\proj i \ea}} {\ta_i}}
  \\
  \infer%[stlc-inj]
  {\dernat \ca \ea {\ta_i}}
  {\dernat \ca {\ob {\inj i \ea}} {\sumfam \ta}}
  \quad
  \infer%[stlc-case]
  {\dernat \ca \ea {\sumfam \ta}\\
   \famx{\dernat {\ca, \ob {\eva_\ii} : {\ta_\ii}} {\ob {\eb_\ii}} \tc}{\ii}}
  {\dernat \ca {\ob {\fdesumfam \ea \eva \eb}} \tc}
  \\
  \infer%[stlc-trivial]
  { }
  {\dernat \ca {\ob \eunit} \tunit}

  \infer%[stlc-var]
  { }
  {\dernat {\ca, \var \ev \ta} \ev \ta}

  \infer%[stlc-absurd]
  {\dernat \ca \ea \tempty}
  {\dernat \ca {\ob {\eabsurd \ea}} \ta}
\end{smallmathparfig}

\begin{smallmathparfig}
  {fig:LCfull-betaeta}
  {$\betaeta$-equivalence for $\LCfull$}

  \app {(\lam \eva \ea)} \eb
  \argred \beta
  \subst \ea \eva \eb

  \annot \ea {\fun \ta \tb}
  \argred \eta
  \lam {\annot \eva \ta} {\app \ea \eva}
  \\
  \proj \ii {\pairfam \ea}
  \argred \beta
  \ea_\ii

  \annot \ea {\prodfam \ta}
  \argred \eta
  \pairlam {\proj \idx \ea}
  \\
  \fdesumfam {\inj \jj \ea} \eva \eb
  \argred \beta
  \subst {\eb_\jj} {\eva_\jj} \ea

  \subst \ea \eva {\eb : \sumfam \ta}
  \argred \eta
  \fdesum \eb {\inj \ii {\evb_\ii}} {\subst \ea \eva {\inj \ii {\evb_\ii}}} \ii
  \\
  \infer
  {\dernat \ca \ea \tunit}
  {\dernat \ca
    {\ea
     \argred \eta
     \eunit}
   \tunit}

  \infer
  {\dernat \ca \eb \tempty\\
   \dernat {\ca, \var \eva \tempty} \ea \ta}
  {\dernat \ca
    {\subst \ea \eva \eb
     \argred\eta
     \eabsurd \eb}
    \ta}
  \\
  \text{Derived rules:}

  \infer
  { }
  {\dernat \ca {\ea_1 \argeq\eta \ea_2} \tunit}

  \infer
  {\dernat \ca \eb \tempty}
  {\dernat \ca {\ea_1 \argeq\eta \ea_2} \ta}
\end{smallmathparfig}

\fullrefnoname{fig:LCfull} gives the grammar and typing rules for the
full simply-typed $\lambda$-calculus $\LCfull$. By symmetry with the
pair projections $\proj \ii \ea$, we use $\inj \ii \ea$ for sum
injection. We use $\famx \dots \iI$ for a family of objects indexed by
$\iI$. The common indexing family, when dealing with syntactic binary
operations, is $\{1,2\}$, and we will most of the time leave it
implicit. Finally, $\fdesumfam \ea \eva \eb$ is a compact syntax for
our full sum-elimination syntax, $\pddesumfam \ea \eva \eb$.

\begin{definition}\label{def:closed-type}
  A \emph{closed} type does not contain atomic types.\Long{\footnote{A
      reviewer remarks that the terminology ``atomic type'' is
      awkward; if ``atom'' means ``indecomposable'', then unit types
      $\tunit$, $\tempty$ could arguably be considered atomic. In
      retrospect, we agree that ``type variable'', as naturally used
      in polymorphic type systems, would be a better terminology, and
      plan to use it in the future.}}{}
\end{definition}

\begin{definition}
  We write $\dernatded \ca \ta$ when $\ea$ exists such that $\dernat \ca \ea \ta$.
\end{definition}

We define $\betaeta$-equivalence as the smallest congruent relation
$\parel{\argeq\betaeta}$ closed by the union of the $\beta$-reduction
relation $\parel{\argred\beta}$ and the $\eta$-expansion relation
$\parel{\argred\eta}$ of \fullrefnoname{fig:LCfull-betaeta}. We have
not explicited the full typing assumptions, but only type-preserving
rewrites are considered. The derived rules are not primitives, they
are derivable from $\eta$-expansion at unit and empty type.

\begin{version}{\Long}
  If $\ea_1, \ea_2$ are of type $\tunit$, then they both rewrite to
  $\eunit$ and are thus equal. The derivation for equality under
  absurdity is less direct: if the current typing context $\ca$ is
  inconsistent, that is there is a derivation $\annot \eb \tempty$,
  then any term of \emph{any} type $\annot \ea \ta$ can be seen as the
  result of the substitution $\subst \ea \ev \eb$ for a variable $\ev$
  that does not appear in $\ea$, and is thus equal to $\eabsurd \eb$;
  in particular, any two terms of the same type are equal.
\end{version}

\begin{theorem}[Strong normalization]
  $\beta$-reduction is strongly normalizing in the full simply-typed
  $\lambda$-calculus $\LCfull$.
\end{theorem}

\begin{theorem}[Confluence]
  $\beta$-reduction is confluent for the full simply-typed
  $\lambda$-calculus $\LCfull$: each term has a unique $\beta$-short
  normal form.
\end{theorem}

\begin{lemma}[Inversion]\label{lem:lc-inversion}
  A closed $\beta$-normal form (in an empty context) starts with an
  %DR: "introduction form" is not defined
  introduction form.
\end{lemma}

\subsection{Contextual equivalence}

A common definition of contextual equivalence, for System F for
example, is that two terms $\dernat \ca {\ea \bocomma \eb} \ta$ are
contextually equivalent if there exists no separating context
$\dernat \emptyset {\plug \Ca {\dernat \ca \ehole \ta}} {\sum \tunit
  \tunit}$ such that $\plug \Ca \ea \argeq\beta \inj 1 \eunit$ and
$\plug \Ca \eb \argeq\beta \inj 2 \eunit$. For a trivial example, if
$\ca \eqdef
(\var \eva {\sum \tunit \tunit},
 \var \evb {\sum \tunit \tunit})$
, then
$\plug \Ca \ehole
\eqdef
\app {\app
  {\plam \eva {\lam \evb \ehole}}
  {\pinj 1 \eunit}}
  {\pinj 2 \eunit}$
separates $\eva$ and $\evb$.

This definition is too weak in presence of atomic types. Consider the
context $\ca \eqdef (\var \eva \xa, \var \evb \xa)$ and the terms
$\dernat \ca {\bo{\eva \bocomma \evb}} \xa$. We want a definition of
contextual equivalence that declares these terms inequivalent, but
there is no distinguishing context in the sense above as we have no
information on $\xa$, and thus no way to provide distinct values for
$\eva$ and $\evb$. The variables $\eva$ and $\evb$ \emph{could} be
distinct depending on what is the unknown type represented by the
abstract type $\xa$.

\begin{definition}\label{def:model}
  A \emph{model} $\ma$ is a mapping from atomic types to closed
  types.
\end{definition}

If $x$ is some syntactic object containing types, we write
$\modapp \ma x$ for the result of replacing each atomic type in $x$ by
its image in the model $\ma$. If $\dernat \ca \ea \ta$ holds, then it
is also the case that
$\dernat {\modapp \ma \ca} \ea {\modapp \ma \ta}$ holds; we may write
the term $\modapp \ma \ea$ as well, to emphasize that we look at its
typing in the model.

\begin{definition}\label{def:conteq}
  If $\dernat \ca {\ea \bocomma \eb} \ta$ and for a given model $\ma$, we say
  that $\ea$ and $\eb$ are \emph{contextually equivalent in $\ma$},
  written $\ea \argeq{\conteqmod \ma} \eb$, if
  \begin{smallmathpar}
    \forall \Ca,\ %
    \dernat \emptyset
      {\plug \Ca {\dernat {\ma(\ca)} \ehole {\ma(\ta)}}}
      {\sum \tunit \tunit}
    \Rightarrow
    \plug \Ca {\ma(\ea)}
    \argeq\beta
    \plug \Ca {\ma(\eb)}
  \end{smallmathpar}%
  We say that $\ea$ and $\eb$ are \emph{contextually equivalent},
  written $\ea \argeq\conteq \eb$, if they are contextually equivalent
  in all models.
\end{definition}

\subsection{Semantic equivalence}

\begin{definition}[Semantics of types]
  For a closed type $\ta$ we define the set of \emph{semantic values}
  of $\ta$, written $\closedtypesem \ta$, by induction on $\ta$ as
  follows:
  \begin{mathpar}
    \begin{array}{l@{\quad\eqdef\quad}l}
      \closedtypesem{\fun \ta \tb}
      & \text{total functions from $\closedtypesem \ta$ to $\closedtypesem \tb$} \\
      \closedtypesem {\prod \ta \tb}
      & \{ (\va, \vb) \mid \va \in \closedtypesem \ta, \vb \in \closedtypesem \tb \} \\
      \closedtypesem \tunit
      & \{ \star \} \\
      \closedtypesem {\sum \ta \tb}
      & \{ (1, \va) \mid \va \in \closedtypesem \ta \}
        \uplus
        \{ (2, \vb) \mid \vb \in \closedtypesem \tb \} \\
      \closedtypesem \tempty & \emptyset \\
    \end{array}
  \end{mathpar}

  For an arbitrary type $\ta$ we write $\typesem \ta \ma$ for
  $\closedtypesem {\modapp \ma \ta}$.
\end{definition}

We remark that $\typesem \ta \ma$ is always a finite type, whose elements
can be enumerated. It is thus decidable whether two elements of
$\typesem \ta \ma$ are equal as mathematical objects; at function
types, one compares two functions pointwise on their finite domain.

\begin{definition}[Semantics of environments]
  For a closed typing environment $\ca$, we define the set
  $\closedenvsem \ca$ of \emph{semantic valuations},
  functions from the domain of $\ca$ to semantic values such that:
  \begin{mathpar}
    \closedenvsem \ca

    \eqdef

    \{ \vca\ \mid\ \forall {\var \eva \ta} \in \ca,\quad
        \vca(\eva) \in \closedtypesem \ta \}
  \end{mathpar}%
  We write $\envsem \ca \ma$ for $\closedenvsem {\modapp \ma \ca}$.
\end{definition}

\begin{definition}[Semantics of typing judgments]
  We write $\closedtypingsem \ca \ta$ for the set of functions from
  semantic valuations of $\ca$ to semantic values in $\ta$:
  $
    \closedtypingsem \ca \ta
    \eqdef
    \closedenvsem \ca
    \to
    \closedtypesem \ta
  $
  .
  We write $\typingsem \ca \ta \ma$ for
  $\closedtypingsem {\modapp \ma \ca} {\modapp \ma \ta}$.
\end{definition}

\begin{definition}[Semantics of terms]
  \label{def:term-semantics}
  For a term $\dernat \ca \ea \ta$ in a judgment with closed types, we
  write $\closedtermsem \ea$ for the set-theoretic semantics of $\ea$,
  as an object of $\closedtypingsem \ca \ta$. The (natural) definition
  is given in full in \cameraversion{the long
    version}{\fullref{def:annex-term-semantics}}, but for example
  \begin{smallmathpar}
    \begin{array}{l@{\eqdef}l}
      \closedtermsem {\lam \ev \ea} (\vca) & (\va \mapsto \closedtermsem \ea (\vca, \ev \mapsto \va))
    \end{array}
  \end{smallmathpar}%
  We write $\termsem \ea \ma$ for $\closedtermsem {\modapp \ma \ea}$.
\end{definition}

\begin{definition}[Semantic equivalence]\label{def:semeq}
  For any terms $\dernat \ca {\ea \bocomma \eb} \ta$ and model $\ma$, we say
  that $\ea$ and $\eb$ are \emph{semantically equivalent in $\ma$},
  written $\ea \argeq{\semeqmod \ma} \eb$, if their semantics are
  (pointwise) equal.
  \begin{smallmathpar}
    \ea \argeq{\semeqmod \ma} \eb
    \ %
    \eqdef
    \ %
    (\forall \vca \in \envsem \ca \ma,\ %
    \termsem \ea \ma (\vca) = \termsem \eb \ma (\vca) \in \typesem \ta \ma)
  \end{smallmathpar}%
  We say that $\ea$ and $\eb$ are \emph{semantically equivalent},
  written $\ea \argeq\semeq \eb$, if they are semantically equivalent
  in all models $\ma$.
\end{definition}

\subsection{Easy relations between equivalences}

\begin{theorem}[$\betaeta$ is semantically sound]
  \label{thm:betaeta-implies-semeq}
  If $\ea \argeq\betaeta \eb$ then $\ea \argeq \semeq \eb$.
\end{theorem}

\begin{theorem}[Semantic equivalence implies contextual equivalence]
  \label{thm:semeq-implies-conteq}
  If $\ea \argeq\semeq \eb$ then $\ea \argeq\conteq \eb$.
\end{theorem}

\begin{version}{\Long}
  \paragraph{On models using closed types \textit{(Long version)}}

  Our definition of a model $\ma$ allows to instantiate atomic types
  with any closed type; this instantiation happens in the world of
  syntax, before we interpret types as sets. It is more common, when
  giving set-theoretic semantics, to instantiate atoms only in the
  semantics, defining models as mapping from atomic types to
  arbitrary \emph{sets}.

  We are fortunate that our grammar of closed types is expressive
  enough to describe any finite set; if we did not have units
  ($\tunit$ and $\tempty$), for example, we could not do this. Our
  more syntactic notion of model can be shared by our definition of
  contextual and semantic equivalence, which is very pleasing.

  Note that, as a consequence, the finite model property is sort of
  built into our definition of the semantic and contextual
  equivalences: we may only distinguish atoms by instantiating them
  with \emph{finite} sets, not arbitrary sets. One may wonder, then,
  whether a notion of semantic equivalence that allows to
  instantiate atoms with \emph{infinite} sets would behave
  differently -- did we really prove the finite model property, or
  just enforce it by stifling our notion of equivalence?

  The coincidence of the finite and infinite interpretations is
  a consequence of the later results of this work on the coincidence
  of semantic and $\beta$-equivalence. If we added the ability to
  instantiate atoms by infinite sets, we would distinguish more: we
  could not prove more terms equal, but would only prove more terms
  \emph{different}. But any pair of terms equal in the finite-set
  semantics is $\betaeta$-equivalent, and
  \fullref{thm:betaeta-implies-semeq} seamlessly extends to infinite
  sets -- those terms must still be equal in an infinite-set
  semantics.
\end{version}

\subsection{Fun-less types and reification}
\label{subsec:funless}

To prove that contextual equivalence implies semantic
equivalence, we build a reification procedure $\reify \va \ma$ that goes
from $\typesem \ta \ma$ to closed terms in $\modapp \ma \ta$ and
satisfies the following properties:
\begin{smallmathpar}
  \forall \va,\ \termsem {\reify \va \ma} \ma = \va

  \forall (\dernat \emptyset \ea \ta),\ %
    \reify {\termsem \ea \ma} \ma \argeq\betaeta \ea
\end{smallmathpar}%

This is difficult in the general case, as it corresponds to
a normalization-by-evaluation procedure -- when you can furthermore
prove that $\reify \va \ma$ is always a normal term. The reification
of finite sums and products is straightforward, but function types are
delicate; intuitively, to reify a function one builds a decision tree
on its input, which requires an enumeration procedure for the input
type~\citep*{altenkirch2004} which may itself be a function, etc.

In the present case, however, the fact that we work with closed types
(no atoms) enables a useful hack: we can eliminate function types by
rewriting them in isomorphic types expressed in
$\LC{(\sysprodu, \syssumu)}$ only. This was inspired by an idea of
Danko Ilik (see for example \citet*{Ilik2015}), which removes sum types
rather than function types. In presence of atomic or infinite types,
neither sum nor function types can be fully removed. In absence of
atoms, function types can be fully removed, but sum types cannot --
there is no type isomorphic to $\tbool$ in $\LC{(\sysfun,\sysprodu)}$.

\fullref{fig:dataty-nofun} defines the fun-less $\dataty \ta$ for each
closed type $\ta$. Its definition is structurally recursive, and uses
an auxiliary definition $\datatyarr {\fun \ta \tb}$ that takes
a function type whose left-hand side $\ta$ is fun-less, and is
structurally recursive on this left-hand side.
We also define pair of transformations $\defun \wild \ta$ from $\ta$
to $\dataty \ta$ and $\refun \wild \ta$ from $\dataty \ta$ to $\ta$,
on both terms and semantic values. On semantic values they are
inverse; on terms they are inverse modulo $\betaeta$-equivalence. The
(natural) definitions are given in full in \cameraversion{the long
  version}{\fullref{def:annex-fun-less-types}}. It uses auxiliary
definitions $\defunarr \wild {\fun \ta \tb}$ and
$\refunarr \wild {\fun \ta \tb}$, and has for example
$\defun \va {\fun \ta \tb}
\eqdef
\defunarr {\vb \mapsto \defun {v (\refun \vb \ta)} \tb} {\fun \ta \tb}$
and
$\defunarr \ea {\fun {\prodfam \ta} \tb}
\eqdef
\defunarr
  {\lam {\eva_1}
    {\defunarr
      {\lam {\eva_2} {\app \ea {\pairfam \eva}}}
      {\fun {\ta_2} \tb}}}
  {\fun {\ta_1} {\pfun {\ta_2} \tb}}$\
.

Finally, we also have that the isomorphisms on semantic values and
ground terms commute.

\begin{smallmathparfig}
  {fig:dataty-nofun}
  {Fun-less data types}
  \dataty {\fun \ta \tb} \eqdef \datatyarr {\fun {\dataty \ta} {\dataty \tb}}

  \dataty {\prodfam \ta} \eqdef \prodlam {\dataty {\ta_\idx}}

  \dataty \tunit
  \eqdef \tunit

  \dataty {\sumfam \ta}
  \eqdef \sumlam {\dataty {\ta_\idx}}

  \dataty \tempty
  \eqdef \tempty
  \\
  \begin{array}{l@{\quad\eqdef\quad}l}
    \datatyarr{\fun {(\prodfam \ta)} \tc}
    & \datatyarr{\fun {\ta_1} {\datatyarr{\fun {\ta_2} \tc}}}
    \\ \datatyarr{\fun \tunit \tb}
    & \tb
    \\ \datatyarr{\fun {(\sumfam \ta)} \tb}
    & \datatyarr{\fun {\ta_1} \tb} \times \datatyarr{\fun {\ta_2} \tb}
    \\ \datatyarr{\fun \tempty \tb}
    & \tunit
  \end{array}
  \\
  \forall \va, \refun {\defun \va \ta} \ta = \va

  \forall \ea, \refun {\defun \ea \ta} \ta \argeq\betaeta \ea

  \forall \ea,\quad
    \defun {\closedtermsem \ea} \ta = \closedtermsem {\defun \ea \ta}
    \ \wedge\ %
    \refun {\closedtermsem \ea} \ta = \closedtermsem {\refun \ea \ta}
\end{smallmathparfig}

\begin{theorem}[Reification]
  \label{thm:reification}
  For each semantic value $\va$ in $\closedtypesem \ta$ we can define
  a closed term $\closedreify \va$ such that
  $\closedtermsem {\closedreify \va} = \va$.
\end{theorem}

\fullref{thm:reification} establishes that semantic inhabitation and
provability coincide at closed types.

\begin{corollary}[Inhabited or inconsistent]
  \label{cor:inhabited-or-inconsistent}
  For $\ca, \ta$ closed, if $\dernatded \ca \ta$ then either
  $\dernatded \emptyset \ta$ or $\dernatded \ca \tempty$.
\end{corollary}

\begin{lemma}
  \label{lem:reification-inverse}
  For any closed term of closed type $\dernat \emptyset \ea \ta$
  we have $\closedreify {\closedtermsem \ea} \argeq\betaeta \ea$.
\end{lemma}

\begin{theorem}[Contextual equivalence implies semantic equivalence]
  \label{thm:conteq-implies-semeq}
  If $\ea \argeq\conteq \eb$ then $\ea \argeq\semeq \eb$.
\end{theorem}

\section{Focusing}
\label{sec:focusing}

Consider the usual description of $\beta$-normal forms in the purely
negative fragment of the simply-typed $\lambda$-calculus,
$\LC{(\xa, \sysfun, \sysprodu)}$.
\begin{smallmathpar}
  \begin{array}{llll}
    \text{values} & \ea
    & \gramdef
    & \lam \eva \ea
      \mid \pairfam \ea
      \mid \eunit
      \mid \ena
    \\
    \text{neutrals} & \ena
    & \gramdef
    & \app \ena \ea
      \mid \proj i \ena
      \mid \eva
  \end{array}
\end{smallmathpar}
Values, the story goes, are a sequence of \emph{constructors} applied to
a neutral, which is a sequence of \emph{destructors} applied to
a variable. It is even possible and easy to capture the set of
$\beta$-short $\eta$-long normal forms by adding a typing restriction
to this grammar, asking for the first neutral term $\ena$ found in
a value to be of atomic type $\annot \ena \xa$.

Unfortunately, adding sum types to this picture shatters it
irreparably. If $\inj i \wild$ is a constructor, it should go in
values, and $\matchwith \wild \dots$ is a destructor, it should be
a neutral term former. Adding $\inj i \ea$ to the grammar of values
seems innocuous, but adding \texttt{match} to neutrals raises
a question: should we ask the branches to be neutrals
$\fdesumfam \ena \eva \ena$ or values $\fdesumfam \ena \eva \ea$?
Neither choices work very well.

Asking branches to be neutrals means that the term
\begin{smallmathpar}
  \dernat
    {\var \eva {\sum \xa \xb}}
    {\pddesum \eva \evb {\inj 2 \evb}
                  \evb {\inj 1 \evb}}
    {\sum \xb \xa}
\end{smallmathpar}
is not a valid normal form, and in fact has no valid normal form! We
cannot force all constructors to occur outside branches, as in this
example we fundamentally need to choose a different constructor in
each branch -- committing to either $\inj 1 \wild$ or $\inj 2 \wild$
before matching on $\eva$ would make us stuck, unable to complete our
attempt with a well-formed term.

On the other hand, letting branches be any value introduces normal
forms that really should not be normal forms, such as
$\proj 1 {\pfdesumlam \eva {\evb_\idx} {\pair \ena {\evb_\idx}}}$, clearly
equivalent, for any value of $\eva$, to the neutral $\ena$.

The solution to this problem comes from logic. Logicians remark that
some inference rules are \emph{invertible} and some are
\emph{non-invertible}. A rule is invertible when, used during
goal-directed proof search, it preserves provability: if the
conclusion was provable (maybe using another rule), then applying this
rule results in premises that are also provable. For example, consider
the implication and disjunction introduction rules:
\begin{smallmathpar}
  \infer
  {\derseq {\ca, \ta} \tb}
  {\derseq \ca {\fun \ta \tb}}

  \infer
  {\derseq \ca {\ta_\ii}}
  {\derseq \ca {\sumfam \ta}}
\end{smallmathpar}
Implication introduction is invertible -- this can be proved by
inverting the rule, showing a derivation of $\derseq {\ca, \ta} \tb$
% DR: "with open premise" is unclear
with open premise $\derseq \ca {\fun \ta \tb}$. Disjunction
introduction is not: if one decides to prove, say, $\ta_1$, one may
get stuck while having chosen to prove $\ta_2$ would have worked. Or
maybe one needs to delay this choice until some hypothesis of the
contexts is explored -- which is the heart of our
$\derseq {\sum \xa \xb} {\sum \xb \xa}$ example.

% DR: this paragraph would benefit from a rewrite
Andreoli's focusing~\citep*{andreoli1992} is a technique to restrict
a logic, make its proof term more canonical, by imposing additional
restrictions based on the invertibility of rules. One easy restriction
is that invertible rules, when they can be applied to a judgment,
should be applied as early as possible. The more interesting
restriction is that when one starts applying non-invertible rules,
focusing forces us to apply them as long as possible, as long as the
formula introduced in premises remain at a type where a non-invertible
rule exists. For a complete reference on focusing in intuitionistic
logic, see \citet*{liang2007}.

In programming terms, the fact that the right implication rule is
invertible corresponds to an \emph{inversion principle} on values:
without loss of generality, one can consider that any value of type
$\fun \ta \tb$ is of the form $\lam \eva \ea$. Any value of type
$\prodfam \ta$ is of the form $\pairfam \ea$. This is strictly true
for closed values in the empty context, but it is true \emph{modulo
  equivalence} even in non-empty contexts, as is witnessed by the
$\eta$-expansion principles. If a value $\annot \ea {\fun \ta \tb}$ is
not a $\lambda$-abstraction, we can consider the equivalent term
$\lam \eva {\app \ea \eva}$.

But it is \emph{not} the case that any value of type $\sum \ta \tb$ is
of the form $\inj \ii \ea$, as our example
$\derseq {\sum \xa \xb} {\sum \xb \xa}$ demonstrated. Inspired by
focusing we look back at our grammar of $\betaeta$-normal forms: it is
not about constructors and destructors, it is about term-formers that
correspond to invertible rules and those that do not. To gracefully
insert sums into this picture, the non-invertible $\inj \ii \wild$
should go into the \emph{neutrals}, and case-splitting should be
a \emph{value}. \citet*{scherer2015} introduce focusing in more
details, and present a grammar of \emph{focused} normal forms is,
lightly rephrased, as follows:
\begin{smallmathpar}
  \begin{array}{lll@{~}l}
    \text{values} & \ea & \gramdef
    & ~ \lam \eva \ea
      \mid \pairfam \ea
      \mid \eunit
      \mid \efa
    \\ & & &
      \mid \fdesumfam \eva \evb \ea
    \\
    \text{choice terms} & \efa & \gramdef
    & ~ \annot \ena \xa
      \mid \letin {\var \eva {\sum \ta \tb}} \ena \ea
      \mid \epa
    \\
    \text{negative neutrals} & \ena
    & \gramdef
    & ~ \app \ena \epa
      \mid \proj i \ena
      \mid \eva
    \\
    \text{positive neutrals} & \epa
    & \gramdef
    & ~ \inj i \epa
      \mid \annot \ea \tna
  \end{array}
\end{smallmathpar}
The type $\tna$ on the last line denotes a \emph{negative} type,
defined as a type whose head connective has an invertible
right-introduction rule: $\fun \ta \tb$ or $\prod \ta \tb$ or
$\tunit$. This means that if the argument of an injection
$\inj \ii \wild$ is itself of sum type, it must be of the form
$\inj \jj \wild$ as well; this enforces the focusing restriction that
non-invertible rules are applied as long as the type allows.

It is interesting to compare this grammar to bidirectional type
systems -- when used to understand canonical forms rather than for
type inference. Focusing generalizes the idea that some parts of the
term structure (constructors) are canonically determined by type
information, while some parts (neutrals) are not. It generalizes
bidirectional typing by taking the typing environment into account as
well as the goal type (variables of sum type are split during the
inversion phase), and refines the application syntax $\app \ena \ea$
into the sharper $\app \ena \epa$ where both sub-terms have a neutral
spine.

Our work builds on this focused representation, easily extended with
an empty type $\tempty$. Our presentation of the type system is
different, farther away from the standard $\lambda$-calculi and closer
to recent presentation of focused systems, by using polarized syntax
for types with explicit shifts -- it clarifies the structure of
focused systems. Instead of distinguishing positive and negative types
based on their head connectives, we define two disjoint syntactic
categories $\tpa$ and $\tna$, with explicit embeddings $\shiftp \tna$,
$\shiftn \tpa$ to go from one to the other. In particular, atoms are
split in two groups, the positive atoms of the form $\xpa$ and the
negative atoms $\xna$ -- there is a global mapping from atoms $\xa$ to
polarities, a given atom is either always positive or always
negative. Sometimes we need to consider either types of a given
polarity or atoms of any polarity; we use $\txpa$ for positive types
or negative atoms, and $\txna$ for negative types of positive atoms.

We present this focused $\lambda$-calculus in
\fullref{fig:focused-lambda-calculus}. The focusing discipline is
enforced by the inference rules which alternate between four different
judgments:
\begin{itemize}
\item The \emph{invertible} judgment
  $\dertyinvpp \cxna \cpa \ea \tna \txpb$ corresponds to invertible
  phases in focused proof search. $\cxna$ is a typing environment
  mapping variables to negative types $\tna$ or positive atoms
  $\xpa$. $\cpa$ contains only positive types; it is the part of the
  context that must be decomposed by invertible rules before the end
  of the phase. The two positions $\tna \betweensuccedent \txpb$ in
  the goal are either formulas or empty $(\emptyset)$, and exactly one
  of them is non-empty in any valid judgment. If the goal is
  a negative formula $\tna$, it has yet to be introduced by invertible
  rules during this phase; once it becomes atomic or positive it moves
  to the other position $\txpb$.
  % DR: explain the $\txpb$ notation
\item The \emph{negative focus} judgment $\dertydown \cxna \ena \tna$
  corresponds to a non-invertible phase focused on a (negative)
  formula in the context.
\item The \emph{positive focus} judgment $\dertyup \cxna \epa \tpa$
  corresponds to a non-invertible phase focused on a (positive)
  formula in the goal.
\item The \emph{choice-of-focusing} judgment
  $\dertyfoc \cxna \efa \txpa$ corresponds to the moment of the proof
  search (reading from the conclusion to the premises) where the
  invertible phase is finished, but no choice of focus has been made
  yet. Focusing on the goal on the right uses a positive neutral to
  prove a positive type -- \Rule{foclc-concl-pos}. Focusing on the
  left uses a negative neutral. If the neutral has a positive type, it
  is \texttt{let}-bound to a variable and the proof continue with an
  invertible phase -- \Rule{foclc-let-pos}. If it has a negative
  atomic type, then it must be equal to the goal type and the proof is
  done -- \Rule{foclc-concl-neg}.
\end{itemize}

\begin{smallmathparfig}
  {fig:focused-lambda-calculus}
  {Cut-free focused $\lambda$-terms}
  \begin{array}{ll@{~}l@{~}l}
  \text{negative types}
  &
  \tna, \tnb & \gramdef &
    \xna, \xnb, \xnc
    \mid \fun \tpa \tna
    \mid \prodfam \tna \mid \tunit
    \mid \shiftn \tpa
  \\
  \text{positive types}
  &
  \tpa, \tpb & \gramdef &
    \xpa, \xpb, \xpc
    \mid \sumfam \tpa \mid \tempty
    \mid \shiftp \tna
    \\
  \end{array}

  \txpa, \txpb \gramdef \tpa, \tpb \mid \xna, \xnb
  \quad
  \txna, \txnb \gramdef \tna, \tnb \mid \xpa, \xpb

  \begin{array}{ll@{~}l@{~}l}
  \text{invertible terms}
  &
  \ea, \eb, \ec & \gramdef &
    \lam \eva \ea
    \mid \eunit
    \mid \pairfam \ea
    \mid \annot \efa \tpa
    \\ & & &
    \mid \eabsurd \eva
    \mid \fdesumlam \eva \eva {\ea_\idx}
  \\

  \text{focusing terms}
  &
  \efa, \efb & \gramdef &
    \letin {\annot \eva \tpa} \ena \ea
    \mid {\annot \ena \xna}
    \mid \epa
  \\
  \text{negative neutrals}
  &
  \ena, \enb & \gramdef &
    {\annot \eva \tna}
    \mid \app \ena \epa
    \mid \proj \ii \ena
  \\
  \text{positive neutrals}
  &
  \epa, \epb & \gramdef &
    \inj \ii \epa
    \mid \annot \eva \xpa
  \end{array}
  \\
  \infer%[foclc-lam]
  {\dertyinvpp{\cxna}{\cp, \var \ev \tpa}{\ea}{\tna}{\emptyset}}
  {\dertyinvpp{\cxna}{\cp}{\lam \ev \ea}{\fun \tpa \tna}{\emptyset}}

  \infer%[foclc-pair]
  {\famx{\dertyinvpp{\cxna}{\cp}{\ea_\ii}{\tna_\ii}{\emptyset}}\ii}
  {\dertyinvpp{\cxna}{\cp}{\pairfam \ea}{\prodfam \tna}{\emptyset}}

  \infer%[foclc-case]
  {\famx{\dertyinvpp{\cxna}{\cp, \var \ev {\tpb_\ii}}{\ea_\ii}{\tna}{\txpa}}{\ii}}
  {\dertyinvpp{\cxna}{\cp, \var \ev {\sumfam \tpb}}
    {\fdesumlam \eva \eva {\ea_\idx}}{\tna}{\txpa}}

  \infer%[foclc-absurd]
  { }
  {\dertyinvpp{\cxna}{\cp, \var \ev \tempty}{\eabsurd \ev}{\tna}{\txpa}}

  \infer%[foclc-trivial]
  { }
  {\dertyinvpp{\cxna}{\cp}{\eunit}{\tunit}{\emptyset}}

  \infer[foclc-inv-foc]
  {\dertyfoc {\cxna, \cxnap} \efa {\mergeformulas{\txpa}{\txpb}}}
  {\dertyinvpp{\cxna}{\shiftpat{\cxnap}}{\efa}{\shiftnat \txpa}{\txpb}}

  \infer[foclc-concl-pos]
  {\dertyup \cxna \epa \tpa}
  {\dertyfoc \cxna \epa \tpa}

  \infer[foclc-concl-neg]
  {\dertydown {\cxna} \ena {\xna}}
  {\dertyfoc {\cxna} \ena {\xna}}
  \quad
  \infer[foclc-let-pos]
  {\dertydown \cxna \ena {\shiftn \tpa}\\
   \dertyinvpp \cxna {\ob \ev : \tpa} \ea {\emptyset} \txpb}
  {\dertyfoc \cxna {\letin \ev \ena \ea} \txpb}

  \infer%[foclc-var-neg]
  { }
  {\dertydown {\cxna, \var \ev \tna} \ev \tna}

  \infer%[foclc-var-pos]
  { }
  {\dertyup {\cxna, \var \ev \xpa} \ev \xpa}

  \infer%[foclc-proj]
  {\dertydown \cxna \ena {\prodfam \tna}}
  {\dertydown \cxna {\proj i \ena} {\tna_i}}

  \infer%[foclc-foc-inv]
  {\dertyinvpp \cxna \emptyset \ea \tna {\emptyset}}
  {\dertyup \cxna \ea {\shiftp \tna}}

  \infer%[foclc-app]
  {\dertydown \cxna \ena {\fun \tpa \tna}\\
   \dertyup \cxna \epa \tpa}
  {\dertydown \cxna {\app \ena \epa} \tna}

  \infer%[foclc-inj]
  {\dertyup \cxna \epa {\tpa_i}}
  {\dertyup \cxna {\inj i \epa} {\sumfam \tpa}}
  \\
  \text{shift-or-atom notations}

  \begin{array}{ll}
    \shiftpat \tna \eqdef \shiftp \tna
    &
    \shiftpat \xpa \eqdef \xpa
    \\
    \shiftnat \tpa \eqdef \shiftn \tpa
    &
    \shiftnat \xna \eqdef \xna
  \end{array}
\end{smallmathparfig}

The notation $\shiftpat \wild$ takes a negative-or-atomic type $\txna$
and returns a positive type. It is used in the rule
\Rule{foclc-inv-foc} that concludes an invertible phase and starts the
choice-of-focusing phase. It may only be applied when the positive
context $\cpa$ is of the form $\shiftpat \cxnap$ for some $\cxnap$,
that is, when it only contains negative or atomic formulas -- it has
been fully decomposed.

Notice that the sum-elimination rule in the invertible judgment
eliminates a variable $\eva$, and not an arbitrary term, and
re-introduces variables with the same name $\eva$, shadowing the
previous hypothesis of sum type: there is no need to refer to it
anymore as we learned its value. This cute trick is not fundamental
for a focused calculus, but it corresponds to the intuition of
the corresponding sequent-calculus rule, and let us actually remove
positive types from the context to have a negative-or-atomic context
at the end of the phase.

For any judgment, for example $\dertyup \cxna \epa \tp$, we use the
version without a term position, for example $\dernatup \cxna \tp$, as
the proposition that there exists a well-typed term:
$\exists \epa,\ \dertyup \cxna \epa \tp$. This is also an invitation
to think of the derivation as a logic proof rather than a typed
program.

% DR: "commuting conversion" is not defined
Our focused terms are \emph{cut-free} in the sense that they contain
no $\beta$-redexes, even modulo commuting conversions. The rule
\Rule{foclc-let-pos} does look like a cut, proving
$\dernatfoc \cxna \txpb$ from $\dertydown \cxna \ena {\shiftn \tpa}$
and $\dertyinvpp \cxna {\var \eva \tpa} \ea \emptyset \txpb$, but
notice that substituting the negative neutral $\ena$ inside the
invertible proof $\ea$ would not create a $\beta$-redex: we know that
$\eva$ is matched over during the invertible phase, but $\ena$ cannot
start with a constructor so $\matchwith \ena \dots$ cannot reduce. If
you are interested in focused systems that do have cuts and
interesting dynamic semantics, then the abstract machine calculi of
\citet*{curien2016} are a better starting point.

\subsection{(Non-)canonicity}\label{subsec:non-canonicity}

When we look at the purely negative fragment of our calculus
$\LC{(\xna, \sysfun, \sysprodu)}$, we can prove that the focused
$\lambda$-terms correspond exactly to the usual notion of
$\beta$-short $\eta$-long normal forms. For example, consider the
valid terms for the judgment
$\dertyinvpp
  {\var \ev {\fun \xpb {\prodfam \xnc}}} \emptyset
  \what
  {\fun \xpb {\prodfam \xnc}} \emptyset
$
. Neither $\eva$ nor $\lam \evb {\app \eva \evb}$, that would be
well-typed for the corresponding un-focused judgment, are valid
according to our inference rules. There is exactly one valid
derivation in our system, for the term
$\lam \evb {\pairlam {\proj \idx {\papp \eva \evb}}}$ which is the
$\eta$-long normal form of $\eva$ at this type.

A consequence of this result is that the focused $\lambda$-calculus is
canonical for the purely negative fragment (or, in fact, the purely
positive fragment): if we have
$\dertyinvpp \cna \emptyset {\ea \bocomma \eb} \tna \xna$ with
$\ea \neq_\alpha \eb$, then $\ea \argneq\betaeta \eb$ and
$\ea \argneq\conteq \eb$ -- these are known to be equivalent in the
negative fragment.

Focusing is not canonical anymore in mixed-polarity settings. The
first source of non-canonicity is that there may a free choice of
ordering of invertible rules in a phase; consider the judgment
$\dertyinvpp \cxna
  {\var \eva {\sumfam \tpa}} \what {\fun \tpb \tna} \emptyset
$
for example, one may either do a case-split on $\sumfam \tpa$ or
introduce a $\lambda$-abstraction for the function type
$\fun \tpb \tna$: $\fdesumlam \eva \eva {\lam \evb \what_\idx}$ or
$\lam \evb {\fdesumlam \eva \eva {\what_\idx}}$ are both valid term
prefixes. This is solved by declaring that we do not care about the
ordering of invertible rules within a single phase.

\begin{definition}
  We define the equivalence relation $\parel{\argeq\icc}$ as allowing
  well-typed permutations of two adjacent invertible rules. For
  example we have
  $
    \eabsurd \eva
    \argeq\icc
    \lam \evb {\eabsurd \eva}
  $.
\end{definition}

From now on any notion of normal form being discussed should be
understood as a \emph{quasi}-normal form, a normal form modulo
invertible commuting conversions $\parel{\argeq\icc}$. This is
a reasonable approximation of the idea of normal form, as it is easily
decidable. Indeed, while in general commuting conversions may relate
very different terms, they can be easily decided inside a single
invertible phase, for example by imposing a fixed ordering on
invertible rules. By definition of invertibility, any ordering
preserves completeness.

The other more fundamental source of non-canonicity is that two
non-invertible phases may be independent from each other, and thus be
ordered in several possible ways, giving distinct but equivalent
terms. For example,
$\letin {\eva_1} {\ena_1} {\letin {\eva_2} {\ena_2} {\pairfam \eva}}$
and
$\letin {\eva_2} {\ena_2} {\letin {\eva_1} {\ena_1} {\pairfam \eva}}$
are equivalent at type $\prod \xna \xnb$ if
$\eva_1 \notin \ena_2, \eva_2 \notin \ena_1$. This source of
redundancy is non-local -- the permutable \texttt{let}-bindings may be
miles away inside the term. It requires a global approach to recover
canonicity, which we discuss in \fullref{sec:saturated-calculus}.

\subsection{Computational completeness}
\label{subsec:computational-completeness}

\label{def:polerase}
We can define a \emph{depolarization} operation $\polerase \wild$ that
takes polarized types and erases polarity information. The definition
is given in full in \cameraversion{the long
  version}{\fullref{def:annex-polerase}}, but for example we have
$\polerase \xna \eqdef \xa$,
$\polerase {\fun \tpa \tna}
 \eqdef
 \fun {\polerase \tpa} {\polerase \tna}
$,
and $\polerase {\shiftp \tna} \eqdef \polerase \tna$.

\label{def:focerase}
This erasure operation can be extended to a \emph{defocusing}
operation $\focerase \wild$ on focused terms that preserves typing modulo depolarization. For
example, if $\dertyinvpp \cxna \cpa \ea \tna \txpb$ holds, then
in the un-focused system
$\dernat
  {\polerase \cxna, \polerase \cpa}
  {\focerase \ea}
  {\mergeformulas {\polerase \tna} {\polerase \txpb}}
$
holds -- with $\mergeformulas \ta \emptyset \eqdef \ta$ and
conversely. This operation is defined on terms as a direct mapping on
all $\lambda$-term formers, except the \texttt{let}-definition form
which does not exist in the unfocused calculus and is substituted
away:
$\focerase {\letin \eva \ena \ea}
 \eqdef
 \subst {\focerase \ea} \eva {\focerase \ena}$.

Going from a focused system to a system with less restrictions is
easy. The more interesting statement is the converse, that for any
un-focused term $\ea$ there exists an equivalent focused term $\eap$.

\begin{theorem}[Completeness of focusing]
  \label{thm:focusing-complete}
  \begin{smallmathpar}
\dernat
  {\polerase \cxna, \polerase \cpa}
  \ea
  {\mergeformulas {\polerase \tna} {\polerase \txpb}}
\\
\implies

\exists \eap,\qquad
  \focerase \eap \argeq \betaeta \ea
  \quad\wedge\quad
  \dertyinvpp \cxna \cpa \eap \tna \txpb
  \end{smallmathpar}
\end{theorem}
\begin{proof}
  Completeness of focusing is a non-trivial result, but it is
  independent to the contributions of the current work, and in
  particular extends gracefully to the presence of an empty type. See
  for example proofs in \citet*{liang2007, ahmad2010, simmons2011}.
\end{proof}

\subsection{Choices of polarization}
\label{subsec:polarization-choices}

We mentioned that a given un-polarized atom $\xa$ must either appear
always positively $\xpa$ or always negatively $\xna$ in our
judgments. Violating this restriction would break completeness, as for
example $\derseq \xpa \xna$ is not provable -- they are considered
distinct atoms. But the global choice of polarization of each atom is
completely free: completeness holds whatever choice is made. Those
choices influence the operational behavior of proof search:
\citet*{chaudhuri-pfenning-price-2008} shows that using the
negative polarization for all atoms corresponds to backward proof
search, whereas using the positive polarization corresponds to forward
proof search.

\begin{version}{\Not\cameraversion}
Similarly, there is some leeway in insertion of the polarity shifts
$\shiftp \wild$ and $\shiftn \wild$; for example, $\sum \tempty \xpa$
and $\sum \tempty {\shiftp {\shiftn \xpa}}$ depolarize to the same
formula, but admit fairly different terms -- the double-shifting
allows an invertible phase to start right after $\pinj 2 \wild$. When
transforming a non-polarized type into a polarized type, two
strategies for inserting shifts are notable. One is to insert as few
shifts as possible; the terms inhabiting the minimally-shifted
judgment are in one-to-one correspondence with terms of the un-focused
system that ``respect the focusing restriction''. The other is to
insert double-shifts under each connective; the terms inhabiting these
double-shifted judgments are in one-to-one correspondence with
unfocused sequent terms -- \citet*{zeilberger2013} relates this to
double-negation translations from classical to intuitionistic logic.
\end{version}

\section{Saturated focused $\lambda$-calculus}
\label{sec:saturated-calculus}

In \fullref{subsec:non-canonicity} we explained that the essential
source of non-canonicity in focused term systems is that distinct
non-invertible phases may be independent from each other: reordering
them gives syntactically distinct terms that are observably equivalent
in a pure calculus. Being such a reordering of another term is
a highly global property, that cannot be decided locally like
invertible commuting conversions.

Logicians introduced \emph{maximal
  multi-focusing}~\citep*{chaudhuri-miller-saurin-2008} to quotient
over those reorderings, and \citet*{scherer2015} expressed this in a
programming setting as \emph{saturation}. The idea of maximal
multi-focusing is to force each non-invertible phase to happen
\emph{as early as possible} in a term, in parallel, removing the
potential for reordering them. However, in general there is no
goal-directed proof search (or term enumeration) procedure that
generates only maximally multi-focused derivations, as one cannot
guess in advance what non-invertible phases will be useful in the rest
of the term -- to introduce them as early as possible. Saturation is
a technique specific to intuitionistic logic: when a non-invertible
phase starts, instead of trying to guess which non-invertible phases
would be useful later, one saturates the context by performing
\emph{all} the possible left-focused phases, \texttt{let}-binding all
the neutrals that might be used in the rest of the term. One can think
of a neutral of positive type $\annot \ena \tpa$ as an
\emph{observation} of the current environment: we are saturating by
performing all possible observations before making a choice -- a right
focusing phase. Note this strategy would be invalid in an effectful
language, or a resource-aware logic where introducing unused
sub-derivations can consume necessary resources and get you stuck.

In general there may be infinitely many distinct observations that can
be made in a saturation phase -- consider the context
$(\var z \xpa, \var s {\fun \xpa {\shiftn \xpa}})$, and a type system
that would enforce complete saturation would then have to admit
infinite terms. Instead, \citet*{scherer2015} relax the definition of
saturation by allowing saturation phases to introduce only a finite
subset of deducible neutrals $\annot \ena \tpa$. They prove that
canonicity holds (in a system without the empty type) in the following
sense: if two saturated terms made the same saturation choices, then
they are equivalent if and only if they are syntactically the same --
modulo $\parel{\argeq\icc}$. In their work, the notion of equivalence
is $\betaeta$-equivalence of the defocused forms.

\subsection{The saturated type system}

The saturated type system is given in
\fullref{fig:saturated-lambda-calculus}. The neutral judgments are
identical to the focused type system of
\fullref{fig:focused-lambda-calculus}, and most of the invertible
rules are also identical. The only change is that the rule
\Rule{sinv-sat} moving from the invertible phase to the focused phase,
instead of merging the two contexts $\cxna; \cxnap$ in a single
context position as \Rule{foclc-inv-foc}, now keeps them separate.

\begin{smallmathparfig}
  {fig:saturated-lambda-calculus}
  {Cut-free saturated focused type system}
  \infer%[sinv-lam]
  {\dertysinv \cxna {\cp, \var \ev \tpa} \ea \tna \emptyset}
  {\dertysinv \cxna \cp {\lam \ev \ea} {\fun \tpa \tna} \emptyset}
  \qquad
  \infer%[sinv-pair]
  {\famx{\dertysinv \cxna \cp {\ea_\ii} {\tna_\ii} \emptyset} \ii}
  {\dertysinv \cxna \cp {\pairfam \ea} {\prodfam \tna} \emptyset}

  \infer%[sinv-trivial]
  { }
  {\dertysinv \cxna \cp \eunit \tunit \emptyset}

  \infer%[sinv-absurd]
  { }
  {\dertysinv \cxna {\cp, \var \ev \tempty} {\eabsurd \ev} \tna \txpb}

  \infer%[sinv-case]
  {\famx{\dertysinv \cxna {\cp, \var \ev {\tpa_\ii}}{\ea_\ii} \tna \txpb}{\ii}}
  {\dertysinv \cxna {\cp, \var \ev {\sumfam \tpa}}
    {\fdesumlam \eva \eva {\ea_\idx}} \tna \txpb}
  \\
  \infer[sinv-sat]
  {\dertysat \cxna \cxnap \efa {\mergeformulas \txpa \txpb}}
  {\dertysinv \cxna {\shiftpat \cxnap} \efa {\shiftnat \txpa} \txpb}

  \infer[sat]
  {(\bar{\ob\en}, \bar{\tp}) \eqdef
   \satselect{\cxna, \cxnap}{\txpb}{\left\{ (\ob\en, \tp) \mid
     \begin{array}{c}
      (\dertysdown {\cxna, \cxnap} \en {\shiftn \tp}) \\
      \wedge\ \reluses {\ob\en} \cxnap
     \end{array} \right\}}
   \\\\
   \dertysinv{\cxna, \cxnap}{\ob{\bar \ev} : \bar{\tp}}\ea \emptyset \txpb
  }
  {\dertysat \cxna \cxnap {\letin {\bar\ev} {\bar\en} \ea} \txpb}
  \\
  \infer[sat-up]
  {\dertysup \cxna \epa \tp}
  {\dertysat \cxna \emptyset \epa \tp}

  \infer[sat-down]
  {\dertysdown \cxna \en \xna}
  {\dertysat \cxna \emptyset \en \xna}
  \\
  \infer%[sat-up-sinv]
  {\dertysinv \cxna \emptyset \ea \tna \emptyset}
  {\dertysup \cxna \ea {\shiftp \tna}}
  \quad
  \infer%[sat-up-atom]
  { }
  {\dertysup {\cxna, \var \ev \xpa} \ev \xpa}
  \quad
  \infer%[sat-down-var]
  { }
  {\dertysdown {\cxna, \var \ev \tna} \ev \tna}
  \\
  \infer%[sat-down-proj]
  {\dertysdown \cxna \ena {\prodfam \tna}}
  {\dertysdown \cxna {\proj i \ena} {\tna_i}}

  \infer%[sat-up-inj]
  {\dertysup \cxna \epa {\tpa_i}}
  {\dertysup \cxna {\inj i \epa} {\sumfam \tpa}}

  \infer%[sat-down-app]
  {\dertysdown \cxna \ena {\fun \tpa \tna}\\
   \dertysup \cxna \epa \tpa}
  {\dertysdown \cxna {\app \ena \epa} \tna}
\end{smallmathparfig}%

This second position $\cxnap$ represents the fragment of the context
that is \emph{new} during the following saturation phase. The
saturation rule \Rule{sat} requires that any introduced neutral $\ena$
use at least one variable of this new context
($\exists \eva \in \cxnap,\ \eva \in \ena$). This guarantees that
a single neutral term cannot be introduced twice by distinct
saturation phases: the second time it will not be new anymore.

% DR: paragraph below is unclear
This new context is also used to know when saturation stops: if an
instance of the \Rule{sat} rule does not introduce any new neutral,
then on the next saturation phase the new context $\cxnap$ will be the
empty context $\emptyset$, allowing saturation to proceed to prove the
goal with one of the two other choice-of-focusing rules.

This aspect of the saturation judgment is reused, unchanged, from
\citet*{scherer2015}. On the other hand, the formulation of the
saturation rule \Rule{sat} is different. We pass the (potentially infinite)
set $\sa$ of new introducible neutrals to a selection function
$\satselect \cxna \txpa \sa$, which returns a finite subset of
neutrals to introduce in a given context.
$\satselect \cxna \txpa \sa$ may not return any subset, we give the
requirement for a selection function to be valid in
\fullref{subsec:selection-function}.

The notation $\letin {\bar \eva} {\bar \ena} \ea$ denotes simultaneous
binding of a (finite) set of neutral terms -- our notion of syntactic
$\alpha$-equivalence is considered to test the (decidable) set
equality.

\subsection{Strong positive neutrals}
\label{subsec:strong-positive-neutrals}

To understand and formalize saturation it is interesting to compare
and contrast the various notions of deductions (seeing our type
systems as logics) at play; how to prove $\ta$ in a context $\ca$?
\begin{itemize}
\item The more general notion of deduction is the unfocused notion of
  proof $\dernatded \ca \ta$ -- proof terms have no restriction. In
  the focused system, it would correspond to looking for a proof of
  an invertible judgment $\dernatinv \emptyset \ca \ta \emptyset$.

\item The neutral judgments $\dertydown \ca \ena \tn$ and
  $\dertyup \ca \epa \tp$ correspond to a less expressive notion of
  ``simple deduction step'', which are iterated by saturation. For
  example, $\dernatup {\shiftn{\sum \xa \xb}} {\sum \xb \xa}$ does
  \emph{not} hold, it requires more complex reasoning than a chain of
  eliminations from the context variables. Focusing decomposes
  a non-focused reasoning into a sequence of such simple deduction
  steps, separated by invertible phases of blind proof search.
\end{itemize}

One notion that is missing is the notion of what is ``already known by
the context''. With the usual non-focused logic, to know if a positive
formula $\tpa$ has been introduced before, we simply check if $\tpa$
is in the context. But the focusing discipline decomposes positive
formulas and removes them from the context.

One could use the judgment $\dernatup \cxna \tp$ instead --
$\dernatup {\xpa, \xpb} {\sum \xpa \xpb}$ as intended.
But $\dernatup \cxna \tp$ is too strong for the purpose of just
retrieving information from the context, as it calls the general
invertible judgment at the end of the focused phase.
$\dernatup \cxna {\shiftp \tna}$ holds whenever $\tna$ is provable
from $\cxna$, not only when $\tna$ is an hypothesis in $\cxna$.

To capture this idea of what can be ``retrieved'' from the context
without any reasoning, we introduce in \fullref{fig:strongup} the
\emph{strong} positive judgment $\dernatstrongup \cxna \tpa$.

\begin{smallmathparfig}
  {fig:strongup}
  {Strong positive judgment $\dertystrongup \cxna \epa \tpa$}
  \infer%[sat-strong-up-neg]
  { }
  {\dertystrongup {\cxna, \var \eva \tna} \eva {\shiftp \tna}}

  \infer%[sat-strong-up-atom]
  { }
  {\dertystrongup {\cxna, \var \eva \xpa} \eva \xpa}

  \infer%[sat-strong-up-inj]
  {\dertystrongup \cxna \epa {\tpa_i}}
  {\dertystrongup \cxna {\inj i \epa} {\sumfam \tpa}}
\end{smallmathparfig}

Strong positive neutrals correspond to the positive patterns of
\citet*{zeilberger-phd}. Those patterns describe the spine of
a non-invertible phase, but they can also characterize invertible
phases: an invertible phase, presented in a higher-order style,
provides a derivation of the goal for any possible positive pattern
passed by the environment. The two following results witness this relation.

\begin{lemma}[Strong decomposition of invertible phases]
  \label{lem:strong-decomposition}
  Consider an invertible derivation
  $\dertysinv \cxna \cpa \ea \tna \txpb$: it starts with invertible
  rules, until we reach a (possibly empty) ``frontier'' of saturated
  subterms $\efa$ to which the rule \Rule{sinv-sav} is applied. Let
  $\famx {\dertysat \cxna {\cxnap_\kk} {\efa_\kk} {\txpbp_\kk}} \kK$
  be the family of such subterms. Then the $\cxnap_\kk$ are exactly
  the contexts such that
  \begin{smallmathpar}
    \forall \tpa \in \cpa,

    \forall \kk,\ \dernatstrongup {\cxnap_\kk} \tpa
  \end{smallmathpar}
\end{lemma}

\begin{lemma}[Strong positive cut]
  \label{lem:strong-cut}~\\
  If both $\dertystrongup \cxna \epa \tpa$ and
  $\dertysinv \cxna \tpa \ea \emptyset \txpb$ hold, then there exists
  a subterm $\efa$ of $\ea$ such that $\dertyfoc \cxna \efa \txpb$ holds.
\end{lemma}

\begin{version}{\Not\cameraversion}
  This result would not be provable for the more expressive judgment
  $\dertyup \cxna \epa \tpa$: there is no obvious way to substitute
  a general $\var \eb \tna$ through $\ea$ that would respect the
  focusing structure -- and return a strict subterm.
  $\dernatstrongup \cxna \tpa \implies \dernatup \cxna \tpa$ is true
  but non-trivial, it relies on the compleness result -- identity
  expansion.
\end{version}

\subsection{Selection function}
\label{subsec:selection-function}

Contrarily to the simpler setting with sums but no empty type, not
all ways to select neutrals for saturation preserve canonicity in
presence of the empty type. Consider for example the focused terms
\begin{smallmathpar}
  \letin \eva {\app {\ob{f}} \eunit}
    {\fdesumlam \eva \eva {\inj 1 ()}}

  \letin \eva {\app {\ob{f}} \eunit}
    {\fdesumlam \eva \eva {\inj 2 ()}}
\end{smallmathpar}%
at the typing
$\dernatded
  {\var f {\fun {\shiftp \tunit} {\shiftn {\sum \tunit \tunit}}},
   \var g {\fun {\shiftp \tunit} {\shiftn \tempty}}}
  {\sum \tunit \tunit}
$.
The set of potential observations is
$\{\app {\ob{f}} \eunit, \app {\ob{g}} \eunit\}$, and both terms made
the same choice of observing only $\app {\ob{f}} \eunit$. The first
term always returns $\inj 1 \eunit$ and the second $\inj 2 \eunit$, so
they are syntactically distinct even modulo $\parel{\argeq\icc}$. Yet
they are $\betaeta$-equivalent as the context is inconsistent. Note
that if $\letin \evb {\app {\ob{g}} \eunit} \wild$ had been introduced
during saturation, the immediately following invertible phase would
necessarily have been $\eabsurd \evb$, and the two terms would thus be
syntactially equal.

To make saturation canonical again, we need a provability completeness
requirement: if there is a possible proof of $\tempty$, we want
saturation to find it. One could cheat, knowing that provability of
any formula is decidable in propositional logic, and test explicitly
for the absence of proof of $\tempty$; but saturation is already doing
proof search\footnote{Saturation synthetizes new sub-terms and can thus
  decide equivalence with $\tempty$, unlike previous methods that
  would only reorder the subterms of the compared terms.}, and
we can extend it gracefully to have this property.

We define our requirement on the selection function in
\fullref{fig:selection-function}. We require that, for any type $\tpa$
that is part of the deducible observations $S$ (by a neutral
$\annot \ena {\shiftn \tpa}$), either $\tpa$ is already retrievable
from the context $\cxna$ (no need to introduce it then) or it is the
type of a neutral $\enap$ selected by the function. We do not require
the same neutral $\ena$ to be selected: there may be infinitely many
different neutrals deducible at $\tpa$, but just having one of them in
the returned set suffices. This definition is not natural, it will be
validated by the results from \fullref{sec:saturation-consistency}.

Note that the types $\tpa$ that can be simply deduced from the context
$\dernatdown \cxna \tpa$ are subformulas of $\cxna$. We know by the
subformula property that they are also subformulas of the root
judgment of the global derivation. In particular, there is only
a finite number of such deducible types $\tpa$ -- this would not hold
in a second-order type system. Valid selection functions exist thanks
to this finiteness.

\begin{smallmathparfig}
  {fig:selection-function}
  {Specification of saturation selection functions}
   \infer[select-specif]
   {\forall \cxna, \sa, \tpa,\\
    \var \ena {\shiftn \tpa} \in \sa
    \quad
    \implies
    \quad
    \dernatstrongup \cxna \tpa
    \;\vee\;
    \exists (\var \enap {\shiftn \tpa}),\; \enap \in \satselect \cxna \tna \sa
   }
   {\satselect \wild \wild \wild \text{ is a valid selection function}}
\end{smallmathparfig}

\begin{version}{\Not\cameraversion}
  Note that two valid selection functions can be merged into a valid
  selection function, by taking the union of their outputs.
\end{version}

\subsection{Completeness of saturation}

Completeness of saturation is relative to a specific choice of
selection function. Indeed, consider the context
\begin{smallmathpar}
\cxna

\eqdef

(\var \eva \xpa, \var \evb {\fun {\shiftp \tunit} {\shiftn \xpa}})
\end{smallmathpar}
In this context, the only deducible positive formula is
$\xpa$, and it is already retrievable from $\cxna$. This means
that a selection function that would satisfy
$\satselect \cxna \xna \sa = \emptyset$ would be a valid selection
function. However, saturating with such a selection function is not
computationally complete: the saturated term $\var \eva \xpa$ has
a valid derivation, but $\letin \evc {\app \evb \eunit} \evc$ does
not -- nor does any equivalent term.

We can order selection functions by pointwise subset ordering:
a function $f$ is above $g$ if it selects at least all of $g$'s
neutrals for each context. A set of saturation functions is
upward-closed if, for any saturation function in the set, any function
above it is in the set.

\begin{theorem}[Completeness of saturation]
  \label{thm:saturation-complete}
  For any focused term $\dertyinvpp \cxna \cpa \ea \tna \txpa$, there
  exists an upward-closed set of selections functions such that
  $\dertysinv \cxna \cpa \eap \tna \txpa$ holds, for a (computable)
  saturated term $\eap$ such that
  $\focerase \ea \argeq\betaeta \focerase \eap$.
\end{theorem}

\begin{version}{\Not\cameraversion}
  Note that, given \emph{two} focused term $\ob{\ea_1}, \ob{\ea_2}$,
  we can merge the selection functions used in the theorem above to
  get a single selection function for which there exists saturated
  terms for both $\ob{\ea_1}$ and $\ob{\ea_2}$.
\end{version}

\section{Saturation consistency}
\label{sec:saturation-consistency}

In this section, we prove the main result of this extension of
saturation to the empty type: if a context is inconsistent, then the
saturation phase will eventually introduce a variable of the empty
type $\tempty$ in the context. This is key to obtaining a canonicity
result -- if saturation sometimes missed proofs of $\tempty$, it could
continue with distinct neutral terms and result in distinct but
equivalent saturated terms.

The informal view of the different ways to deduce a positive formula
presented in \fullref{subsec:strong-positive-neutrals} (general proof,
simple deduction, retrieval from context) gives a specification of
what saturation is doing. From a high-level or big-step point of view,
saturation is trying all possible new simple deductions iteratively,
until all the positives deducible from the context have been added to
it. The following characterization is more fine-grained, as it
describes the state of an intermediary saturation judgment
$\dertysat \cxna \cxnap \efa \txpa$.

The characterization is as follows: any formula that can be ``simply
deduced'' from the old context $\cxna$ becomes ``retrievable'' in the
larger context $\cxna, \cxnap$. This gives a precise meaning to the
intuition that $\cxna$ is ``old''. What we mean when saying that
$\cxnap$ is ``new'' can be deduced negatively: it is the part of the
context that is still fresh, its deductions are not stored in the
knowledge base yet.

\begin{theorem}[Saturation]
  \label{thm:saturation}
  If a saturated proof starts from a judgment of the form
  $
    \dertysat \emptyset {\cxna_0} \efa \txpb
  $
  or
  $
    \dertysinv \emptyset {\cpa_0} \ea \tna \txpb
  $
  then for any sub-derivation of the form
  $
    \dertysat \cxna \cxnap \efa \txpb
  $
  we have the following property:
  \begin{mathpar}
    \forall \tpa,

    \dernatdown \cxna {\shiftn \tpa}

    \implies

    \dernatstrongup {\cxna, \cxnap} \tpa
  \end{mathpar}
\end{theorem}

\begin{definition}
  \label{def:saturated-context}
  $\cxna$ is \emph{saturated} if $\dernatdown \cxna {\shiftn \tpa}$
  implies $\dernatstrongup \cxna \tpa$.
\end{definition}

\begin{corollary}[Saturation]
  \label{cor:saturation}
  If a saturated proof starts from a judgment of the form
  $
    \dertysat \emptyset {\cxna_0} \efa \txpb
  $
  or
  $
    \dertysinv \emptyset {\cpa_0} \ea \tna \txpb
  $
  then for any sub-derivation of the form
  $
    \dertysat \cxna \emptyset \efa \txpb
  $
  the environment $\cxna$ is saturated.
\end{corollary}

\begin{lemma}[Saturated consistency]
  \label{lem:saturated-consistency}
  If $\cxna$ is saturated, then $\ndernatded \cxna \tempty$.
\end{lemma}

\begin{theorem}[Inconsistent canonicity]
  \label{thm:inconsistent-canonicity}
  If $\dernatded \cxna \tempty$, then for any $\efa, \efap$ such that
  $\dertysat \emptyset \cxna {\efa \bocomma \efap} \txpa$ we have
  $\efa \argeq\icc \efap$.
\end{theorem}

\section{Canonicity}
\label{sec:canonicity}

In this section we establish the main result of this article. If two
saturated terms $\dertysinv \cxna \cpa {\ea \bocomma \eap} \tna \txpb$
are not syntactically equivalent ($\ea \argneq\icc \eap$), then there
exists a model $\ma$ in which a context distinguishes $\ea$ from
$\eap$: they are not contextually equivalent.

(We build distinguishing contexts in the un-focused
$\lambda$-calculus, so technically we are distinguishing the defocused
forms $\focerase \ea, \focerase \eap$; the proof crucially relies on
the saturated structure of its inputs, but the code we generate for
computation and separation is more easily expressed unfocused.)

\subsection{Sketch of the proof}

\paragraph{Intuition}

It is helpful to first get some intuition of what a pair of
syntactically distinct normal forms looks like, and what the
corresponding distinguishing context will look like. Suppose we have
$\dertysinv \cxna \cpa {\ea \argneq\icc \eap} \tna \txpb$. We can
explore $\ea$ and $\eap$ simultaneously, until we find the source of
their inequality.

The source of inequality cannot be in an invertible phase, given that
the term formers in invertible phase are completely determined by the
typing (modulo invertible commuting conversions); for example, if
$\tna$ is $\prodfam \tna$, we know that $\ea$ is of the form
$\pairfam \ea$, and $\eap$ of the form $\pairfam \eap$, with
$\ea_\ii \argneq\icc \eap_\ii$ for some $\ii$ -- so we can continue
exploring $\ea_\ii \argneq\icc \eap_\ii$. Same thing if the term
starts with a sum elimination (modulo $\parel{\argeq\icc}$ one can
assume that they eliminate the same variable),
$\fdesumlam \eva \eva {\ea_\idx}
\argneq\icc
\fdesumlam \eva \eva {\eap_\idx}
$
: the subterms $\ea_\ii, \eap_\ii$ in at least one of the
two branches differ.

Similarly, the source of inequality cannot be in the saturation phase,
where both terms saturate on neutrals that are completely determined
by the typing context -- and the saturation selection function -- they
are both of the form $\letin {\bar \eva} {\bar \ena} \wild$ for the
same set of neutrals on each side. The end of this saturation phase is
also type-directed, so both terms stop saturating (they get an empty
$\cxnap$ context) at the same time. The difference must then be in
the neutrals used in the \Rule{sat-up} or \Rule{sat-down} rules,
$\ena \argneq\icc \enap$ or $\epa \argneq\icc \epap$. Note that we
cannot get a positive neutral on one side and a negative neutral on
the other, as usage of those rules is directed by whether the goal
type is a negative atom $\xna$ or a positive type $\tpa$.

Now, two neutrals $\ena \argneq\icc \enap$ or $\epa \argneq\icc \epap$
may differ because their \emph{spine} differ, or because their
sub-terms that are outside the non-invertible phase differ. In the
latter case, finding the source of inequality is a matter of
traversing the common structure towards the sub-terms that differ. The
former case is more interesting -- this pair of neutrals with distinct
spines is what we call \emph{source of inequality}.

In the positive neutral case, we end up on either $\inj \ii \epa$ and
$\inj \jj \epap$ with $\ii \neq \jj$, or distinct variables
$\eva \neq \evb$ of atomic type $\xpa$. In the negative neutral case,
one may encounter distinct neutrals with distinct structure, for
example $\app \ena \epa \neq \eva \neq \proj \ii \enb$ at the same
type $\tna$; negative neutrals should be looked ``upside down'', as in
System L~\citep*{curien2016}: either their head variables differ, or
the same variable is applied a different sequence of elimination
rules.

In the easy case where the source of inequality is a sum constructor
$\pinj \ii \wild \argneq\icc \pinj \jj \wild$, obtaining
a distinguishing context looks relatively easy: we need a context
$\plug \Ca \ehole$ that corresponds to the term traversal we performed
to reach this source of inequality. For example, if we had
$\pairfam \ea \argneq\icc \pairfam \eap$ because $\ea_2 \argneq\icc \eap_2$,
the context fragment corresponding to this reasoning step would be
$\proj 2 \ehole$. This is trickier in the sum-elimination case: if we have
$\fdesumlam \eva \eva {\ea_\idx}
\argneq\icc
\fdesumlam \eva \eva {\eap_\idx}
$
then we need our context to instantiate the variable $\eva$ with the
right value $\inj \ii \epa$ so that the branch we want is taken -- the
one with $\ea_\ii \argneq\icc \eap_\ii$. This is easy if $\eva$ is
a formal variable introduced by a $\lambda$-abstraction: at the point
where our context needs to distinguish the two $\lambda$-abstraction
$\lam \eva \ea \argneq\icc \lam \eva \eap$, we need to use an
application context of the form $\app \ehole {\pinj \ii \epa}$. But
$\eva$ may have been introduced by a left-focusing step
$\letin \eva \ena \ea \argneq\icc \letin \eva \ena \eap$; then we need
to instantiate the variables of the observation $\ena$ \emph{just so}
that we get the desired result $\inj \ii \eap$. When the source of
inequality is on negative neutrals with heads $\eva : \tna$,
$\evb : \tnb$ or positive variables $\eva \neq \evb : \xpa$, we need
the distinguishing context to pass values in the same way to get to
this source of inequality, and also to instantiate the variables
$\eva, \evb$ to get an inequality. If those variables are at an atomic
type, we must pick a model that replaces this atomic type by a closed
type, to instantiate them by distinguishable values at this closed
type.

\paragraph{Positive simplification}

In order to simplify the following arguments, we will suppose that the
context and types of the two saturated proof to distinguish do not use
negative atoms, only positive atoms. In other words, the results only
holds for the focused system in the fragment
$\LC{(\xpa,\sysfun,\sysprodu,\syssumu)}$.

This is a perfectly reasonable simplification in view of our goal,
which is to distinguish inequivalent non-focused term that have
distinct saturated normal forms: we know that any choice of
polarization for the non-focused atoms preserves the existence of
normal forms, so we can make them all positives -- see
\fullref{subsec:polarization-choices}.

In particular, the source of inequality (distinct neutrals whose spine
differs) is always a pair of neutrals $\epa \argneq\icc \epap$, who
contain a syntactic difference
($\inj \ii \wild \argneq\icc \inj \jj \wild$ with $\ii \neq \jj$, or
$\eva \neq \evb : \xpa$) before the end of the non-invertible
phase. Negative neutrals are only used during saturation.

\paragraph{Neutral model}

If $\sa$ is a finite set, let us write $\finty{\sa}$ for the type
$\sum \tunit {\sum \dots \tunit}$ that is in bijection with $S$ (same
number of elements), witnessed by embeddings
$\infin \wild : \sa \to \finty\sa$ and
$\outfin \wild : \finty\sa \to \sa$ such that
\begin{smallmathpar}
  \outfin {\infin x} = x \in \sa
  \uad\wedge\uad
  \dernat \emptyset {\infin {\outfin \ea} \argeq\betaeta \ea} {\finty\sa}
\end{smallmathpar}%
For any two distinct elements $x \neq y \in S$ there exists
a distinguishing context for $\infin x \argneq\conteq \infin y$.

\begin{definition}[Neutral model]
  \label{def:neutral-model}
  Given two syntactically distinct saturated terms
  $\dertyinvpp \cxna \cpa {{\eanz} \argneq\icc {\eanzp}} \tna \txpa$
  (with positive atoms only) with a source of inequality of the form
  \begin{smallmathpar}
  \dertyup \cxnap {\epa \argneq\icc \epap} \tpa
  \end{smallmathpar}
  we define the
  \emph{neutral model} $\nmod {\eanz} {\eanzp}$ (or just $\nmodlight$)
  by
  \begin{smallmathpar}
    \nmodea (\xpb)
    \uad
    \eqdef
    \uad
    \finty{\{ \eva \mid \annot \eva \xpb \in \cxnap \}}
  \end{smallmathpar}%
\end{definition}

We say that $\nmodlight(\xa)$ contains a \emph{code} for each atomic
variable bound at the source of inequality.

\paragraph{Distinguishing outline}

The general idea of our distinguishing context construction is to
instantiate variables just so that each variable of atomic type
$\annot \eva \xpa$ evaluates to its code $\infin \eva$. Thus, when we
end up on distinct neutrals $\epa \argneq\icc \epap$, we know that our
context will send them to distinguishables values.

There are two moments where building a distinguishing context requires
synthesizing a closed value of a type: to instantiate the open
variables in the context of the two terms, and when distinguishing
invertible terms at a function type $\fun \tpa \tna$, which we know to
be of the shape $\lam {\annot \eva \tpa} {\annot \wild \tna}$.

Synthesizing a value for a variable of atomic type $\annot \eva \xpa$
is obvious, we just pick $\infin \eva$ -- this guarantees that, under
this context, $\eva$ will reduce to $\infin \eva$ as expected. For
a variable of sum type $\sum \tpa \tpb$, we have to choose the value
to make sure that the correct branch of the two terms (the one
containing the source of inequality) will be explored, as previously
explained. For a variable $\eva$ of negative type $\tna$, we have to
make sure that any observation of $\eva$, any neutral term $\ena$
whose head variable is $\eva$ (in the syntax of \citet*{curien2016}
they are the $\machine \eva S$), will reduce to the value we need: if
we have $\letin {\var \evb \tpa} \ena \dots$, the neutral $\ena$
should reduce to $\infin \evb$. In other words, much in the spirit of
higher-order focusing \citet*{zeilberger-phd}, we specify the
instantiation of $\annot \eva \tna$ by a mapping over all the
observations over $\eva$ that are made in $\eanz, \eanzp$.

\paragraph{Example}

Consider for example:
\begin{smallmathpar}
\dernat {\var n {\fun {\psum \tunit \xpa} {\shiftn \xpa}}}
  {\begin{array}{c}
     \bletin z
       {\app n {\pinj 1 \eunit}}
       {\letin o {\app n {\pinj 2 z}}
               z}
     \\
     \argneq\icc
     \\
     \bletin z
       {\app n {\pinj 1 \eunit}}
       {\letin o {\app n {\pinj 2 z}}
               o}
   \end{array}}
  \xpa
\end{smallmathpar}

% DR: I found the explanations below insufficient, while
% they seem to be the key step
The shared context in this example is
\begin{smallmathpar}
\letin z
  {\app n {\inj 1 \eunit}}
  {\letin o {\app n {\pinj 2 z}} \ehole}
\end{smallmathpar}%
and the source of inequality is $\ob{z \argneq\icc o}$. The atomic
variables in the context at this point are $\ob{z}$ and $\ob{o}$, so
we have $\nmodlight(\xpa) = \finty{\{\ob z, \ob o\}}$. We have to
provide a value for the negative variable $\ob n$; its
``observations'', the arguments it is called with, are $\inj 1 \eunit$
and $\inj 2 z$, so we define it on these inputs following the general
scheme:
\begin{smallmathpar}
  \ob{\hat n} \eqdef \left\{
    \begin{array}{lll}
      \inj 1 \eunit & \mapsto & \infin {\ob z} \\
      \inj 2 {\infin {\ob z}} & \mapsto & \infin {\ob o} \\
    \end{array}
  \right.
\end{smallmathpar}

The value of $\ob{\hat n}$ on the last element of
$\modapp \nmodlight {\sum \tunit \xpa}$, namely
$\inj 2 {(\infin {\ob o})}$, is not specified; the return type is
inhabited so a value can be chosen, and the specific choice does not
matter.

It is easy to check that the context
$\plug \Ca \ehole \eqdef
 \ob{\app {\plam n \ehole} {\hat n}}
$
is such that plugging both terms will result in distinct closed
values of $\modapp \nmodlight \xpa$, namely
$\infin {\ob z}$ and $\infin {\ob o}$.
From there, building a distinguishing context returning a boolean is
trivial.

\paragraph{Technical challenge}

We outlined the general argument and demonstrated it on an
example. Unfortunately, we found that scaling it to a rigorous,
general proof is very challenging.

When we define instantiation choices for negative types as mappings
from observations to their results, we implicitly rely on the fact
that the observations are distinct from each other. This is obvious
when the domain of these observations is made of first-order datatypes
(no functions), but delicate when some of those observations have
function types -- consider observations against
$\annot \eva {\fun {\shiftp {\fun \tpa \tna}} {\shiftn \xpa}}$.

The natural idea is to inductively invoke the distinguishability
result: if
$\app \eva {\plam \evb \ea} \argneq\icc \app \eva {\plam \evb \eap}$,
then $\ea \argneq\icc \eap$ are distinguishable and $\hat \eva$ can
distinguish those two arguments by passing the right instantiation for
$\evb$. However, making this intuition precise gets us into a quagmire
of self-references: to define instantiation of the variable $\eva$, we
may need to instantiate any such variable $\evb$ whose scope it
dominates, but the argument for distinguishability is really on the
values that $\lam \evb \ea$, $\lam \evb \eap$ have reduced to by the
time they are passed to $\hat\eva$; but the natural way to denote
those values makes a reference to the instances passed for all the
variables in their scope, $\eva$ included...

We are convinced that there is a general inductive argument to be
found that would play in a beautiful way with general polarized type
structure. We have not found it yet, and shall for now use (a focused
version of) the function-removal hack from \fullref{subsec:funless}.

\subsection{Saturated inequivalence}

We argued that if we have $\ea \argneq\icc \eap$, then those terms
must be of the form $\plug \Cb \ena, \plug \Cb \enap$ or
$\plug \Cb \epa, \plug \Cb \epap$ where the neutrals
$\ena \argneq\icc \enap : \xna$ or $\epa \argneq\icc \epap : \tpa$
have distinct \emph{spines} (invertible phases), that is, they differ
even when ignoring their invertible sub-terms. With the simplifying
assumption that all atoms are positively polarized, only the case
$\epa \argneq\icc \epap$ may arise.

This structure is crucial in the proof of canonicity, so we introduce
in \fullref{fig:derneq} a precise inductive definition of this
decomposition as a \emph{saturated inequivalence} judgment
\begin{smallmathpar}
  \derneqinv \ca \cca
    {\plug \Cb {\derneqspine \cap \ccap \epa \epap \tpa}}
    \tna \txpa
\end{smallmathpar}
which represents constructive evidence of the fact that
$\plug \Cb \epa \argneq\icc \plug \Cb \epap$, where $\epa, \epap$ are
the source of inequality as witnessed by the new judgment
$\derneqspineup \cap \epa \epap \tpa$.

The new structure $\cca$ is a \emph{constraint environment}, a list of
equalities of the form $\epa = \eva : \tpa$ or $\eva = \ena : \tpa$
that correspond to knowledge that was accumulated during the
traversal of $\Cb$: under a binding $\letin \eva \ena \ea$ we remember
$\eva = \ena$, and when doing a case-split on a variable
$\eva : \sumfam \tpa$ we remember which branch leads to the source of
inequality by $\inj i \evap = \eva$. Together, $\ccquotient \cap \ccap$
form a \emph{constrained environment} as used in
\citet*{fiore1999}. Note that when decomposing the variable
$\var \eva {\sumfam \tpa}$ we mention it in the constraint
environment, so we also keep it in $\cap$ for scoping purposes: we use
general contexts with types of any polarity, not just negative or
atomic contexts $\cxna$.

Two side-conditions in this judgment capture the essential properties
of saturated terms required to obtain canonicity. In the saturation
case, the condition $\forall \ena \in \bar\ena,\ \ena \notin \cca$
enforces that \texttt{let}-bindings in the constraint environment
$\cca$ all bind distinct neutrals. At the source of inequality, the
condition $\ndernatded \ca \tempty$ enforces that the context is
consistent.

\begin{smallmathparfig}
  {fig:derneq}
  {Saturated inequivalence judgment}
  \text{constraint environment}

  \cca \gramdef \emptyset \mid \cca, \epa = \eva \mid \cca, \eva = \ena
  \\
  \fbox{$\derneqinv \ca \cca
           {\plug \Cb {\derneqspine \cap \ccap \epa \epap \tpap}} \tna \txpb$}
  \\
  \infer
  {\bderneqinv
    {\ca, \var \eva {\sumfam \tpa}, \var {\eva_\ii} \tpa_\ii}
    {\cca, \inj \ii {\eva_\ii} = \eva}
    {\plug \Cb {\derneqspine \cap \ccap \epa \epap \tpap}}
    \tna \txpb}
  {\bderneqinv {\ca, \var \eva {\sumfam \tpa}} \cca
    {\plug
      {\left({\bcasenoinj \eva
        {\inj \ii {\eva_\ii}} \Cb
        {\inj \jj {\eva_{\jj \neq \ii}}} \ea}\!\!\!\!\right)}
      {\derneqspine \cap \ccap \epa \epap \tpap}
    }
    \tna \txpb}

  \text{(other cases omitted for space)}
  \\
  \fbox{$\derneqfoc \ca \cca
          {\plug \Cb {\derneqspine \cap \ccap \epa \epap \tpap}} \txpb$}
  \\
  \infer
  {\forall \ena \in {\bar\ena},\ \ena \notin \cca\\
   \derneqinv
    {\ca, \var {\bar \eva} {\bar \tpa}}
    {\cca, \bar \eva = \bar \ena}
    {\plug \Cb {\derneqspine \cap \ccap \epa \epap \tpap}}
    \emptyset \txpb
  }
  {\derneqfoc \ca \cca
    {\plug {\pletin {\annot {\bar \eva} {\bar \tpa}} {\bar \ena} \Cb}
      {(\derneqspine \cap \ccap \epa \epap \tpap)}}
    \txpb}

  \infer
  {\derneqspineup \ca \epa \epap \tpa\\ \ndernatded \ca \tempty}
  {\derneqfoc \ca \cca
    {\plug \ehole {\derneqspine \ca \cca \epa \epap \tpa}} \tpa}
  \quad
  \text{(other cases omitted for space)}
  \\
  \fbox{$\derneqspineup \ca \epa \epap \tpa$}
  \\
  \infer
  {\eva \neq \evb}
  {\derneqspineup {\ca, \var \eva \xpa, \var \evb \xpa} \eva \evb \xpa}
  \quad
  \infer
  {(\ii \neq \jj)
   \uad\vee\uad
   \derneqspineup \ca \epa \epap {\tpa_{\ii = \jj}}}
  {\derneqspineup \ca {\inj \ii \epa} {\inj \jj \epap} {\sumfam \tpa}}

  \text{no case for $\derneqspineup \ca \ea \eap {\shiftp \tna}$}
\end{smallmathparfig}

\begin{definition}
  \label{def:valid-constraints}
  A constrained environment $\ccquotient \ca \cca$ is \emph{valid} if
  \begin{smallmathpar}
    \ena \in \cca \implies \dertydown \ca \ena \cca

    \epa \in \cca \implies \dertystrongup \ca \epa \cca

    (\epa = \eva) \in \cca, (\epap = \evb) \in \cca, \eva =_\alpha \evb
    \implies \epa =_\alpha \epap

    (\eva = \ena) \in \cca, (\evb = \enap) \in \cca, \ena =_\alpha \enap
    \implies \eva =_\alpha \evb
  \end{smallmathpar}
\end{definition}

\begin{lemma}[Saturated inequivalence]
  \label{lem:saturated-inequivalence}
  ~\\
  If $\dertyinvpp \emptyset \cpa {\ea \argneq\icc \eap} \tna \txpb$
  with positive atoms only,
  then there exists ${\plug \Cb \ehole}, \epa, \epap$ and a valid $\cca$ such that
  \begin{smallmathpar}
    \derneqinv \cpa \emptyset
      {\plug \Cb {\derneqspine \cxnap \cca \epa \epap \tpa}}
      \tna \txpb
    \\
    \ea \argeq\icc \plug \Cb \epa

    \eap \argeq\icc \plug \Cb \epap
  \end{smallmathpar}
  Let us write $\derdecompto \ea \eap \Cb \epa \epap$ when this
  relation holds.
\end{lemma}

\subsection{Focused fun removal}

In this section, we describe how to transform any pair of distinct
saturated terms $\eanz \argneq\icc \eanzp$ in
$\LC{(\xpa, \sysfun, \sysprodu, \syssumu)}$ into a pair of distinct
saturated terms in the restricted type system
$\LC{(\sysprodu, \syssumu)}$. We know that we can apply the
neutral model $\nmod{\eanz}{\eanzp}$ to get a pair of terms without
atoms -- $\LC{(\sysfun,\sysprodu,\syssumu)}$ -- and then the
bijections of \fullref{subsec:funless} give us terms in
$\LC{(\sysprodu,\syssumu)}$. But the impact of these transformations
on the focused and saturated term structure must be studied carefully.

\paragraph{Reduction to closed types}

When applying the neutral model $\nmod{\eanz}{\eanzp}$, we turn atomic
types $\xpa$ into strictly positive types $\modapp \nmodlight \xpa$ of the form
$\sum \tunit {\sum \dots \tunit}$. Consider a binding site of the form
$\plam {\annot \eva \xpa} \ea$, which becomes
$\lam {\annot \eva {\sum \tunit {\sum \dots \tunit}}} {\modapp \nmodlight \ea}$
after transformation. To recover a valid focused term, we need to
insert a big sum elimination to remove the positive variable $\eva$
from the context:
\begin{mathpar}
  \lam \eva
  {\fdesum \eva {\inj \ii \eva}
    {\subst \ea \eva {\inj \ii \eunit}} {\ii \in \modapp \nmodlight \xpa}}
\end{mathpar}
(the family notation here denotes a cascade of sum-eliminations, with
as many cases in total as elements in $\modapp \nmodlight \xpa$). We
also perform such an elimination on each variable in context at the
root.

The substitution of $\inj \ii \eunit$ for $\annot \eva \xpa$ in $\ea$
may not create $\beta$-redexes, as a variable $\eva$ of atomic type
may not occur in eliminated-sum position. By inspection of the typing
rules of our focused system, such a variable can only appear in
a judgment of the form $\dertyup \cxna \eva \xpa$; the focused structure
is preserved by replacing it by the derivation
$\dertyup \cxna {\inj \ii \eunit} {\sum \tunit {\sum \dots \tunit}}$.

For a focused term $\ea$, let us simply write $\modapp \nmodlight \ea$
for this transformed focused term, splitting over each variable of
type $\modapp \nmodlight \xpa$.

The \emph{saturated} structure, however, is not preserved by this
change. The problem is that replacing a new variable $\eva$ by
a closed term $\inj \ii \eunit$ may break the condition that only
``new'' neutrals are introduced during a saturation phase: any neutral
that would only use $\eva$ as the new variable introduced by the last
invertible phase is not new anymore.

However, we can show that the \emph{saturated inequivalence} structure
is preserved by this transformation. There is a shared context to
a source of inequality in the transformed terms, where the
\texttt{let}-bindings might not be new at introduction time, but they
are not redundant, as requested by the inequivalence structure. This
is the property of saturated proofs (besides saturation consistency)
that our distinguishability result relies on.

\begin{theorem}[Inequivalence in the model]
  \label{thm:inequivalence-in-the-model}
  Suppose we have saturated forms $\eanz, \eanzp$
  with positive atoms only. Let us define
  $\eano \eqdef \modapp {\nmod{\eanz}{\eanzp}} \eanz$ and
  $\eanop \eqdef \modapp {\nmod{\eanz}{\eanzp}} \eanzp$.

  If $\derdecompto \eanz \eanzp {\Cb_0} {\epanum{0}} {\epapnum{0}}$,
  then there exists $\Cb_1$, $\epanum{1}$, $\epapnum{1}$ such that
  $\derdecompto \eano \eanop {\Cb_1} {\epanum{1}} {\epapnum{1}}$
\end{theorem}

\paragraph{Function elimination} To turn
$\eano \argneq\icc \eanop$ into terms at function-less types, we use
the transformations presented in \fullref{fig:focused-nofun}. They
correspond to focused versions of the term transformation
$\defun \wild \ta$ and $\defunarr \wild \ta$ of
\fullref{fig:dataty-nofun}, in two interdependent ways: their
definition assumes the transformed terms respect the focused
structures, which gives us rich information on term shapes, and their
result remain in focused form. The specification of these
transformations is as follows:
\begin{smallmathpar}
  \dertyinvpp \cxna \cpa \ea \tna \txpa
  \Rightarrow
  \dertyinvpp {\dataty \cxna} {\dataty \cpa}
    {\defun \ea {\mergeformulas \tna \txpa}}
    {\dataty \tna} {\dataty \txpa}
  \\
  \dertyinvpp \cxna \cpa \ea \tna \txpa
  \Rightarrow
    {\defunarr \ea {\mergeformulas \tna \txpa}}
    {\datatyarr \tna} {\datatyarr \txpa}
  \\
  \infer
  {\dertydown \cxna \ena \tna\\ \dots}
  {\dertydown \cxna \enb \tnap}
  \Rightarrow
  \exists \enbp,
  \infer
  {\dertydown {\datatyarr \cxna} {\defunarr \ena \tna} {\datatyarr \tna}\\ \dots}
  {\dertydown {\datatyarr \cxna} \enbp \tnap}
\end{smallmathpar}
The transformation of negative neutrals is more elegantly expressed in
the sequent-calculus syntax of \citet*{curien2016}: we transform
a command $\machinety \ena S \tna : \dernatded \cxna \tnb$ cutting on
a type $\tna$ into a command
$\defunarr {\machinety \ena S \tna} \tna : \dernatded {\datatyarr
  \cxna} \tnb$ cutting on $\datatyarr \tna$. We will write
$\plug \enb \ena$ when the neutral $\enb$ has head $\ena$ -- it is of
the form $\machine \ena S$.

\begin{smallmathparfig}
  {fig:focused-nofun}
  {Fun-less focused forms}

  \defun \ea \ta

  \text{applies $\defunarr \wild \ta$ on all subterms}
  \\
  \defunarr {\lam \eva {\bdesumlam \eva \eva {\ea_\idx}}} {\fun {\psumfam \tpa} \tna}
  \eqdef
  \vpairlam
    {\defunarr
      {\lam {\eva} {\ea_\idx}}
      {\fun {\tpa_\idx} \tna}}

  \app {\defunarr \ena {\fun {\psumfam \tpa} \tna}} {\pinj \ii \epa}
  \eqdef
  \app {\defunarr {\proj \ii \ena} {\fun {\tpa_\ii} \tna}} \epa
  \\
  \defunarr {\lam \eva {\eabsurd \eva}} {\fun \tempty \tna}
  \eqdef \eunit

  \app {\defunarr \ena {\fun \tempty \tna}} \epa \text{ impossible}
  \\
  \defunarr {\lam \eva \ea} {\fun {\shiftp \tunit} \tna}
  \eqdef \ea

  \app {\defunarr \ena {\fun {\shiftp \tunit} \tna}} \eunit
  \eqdef
  \ena
  \\
  \defunarr {\lam \eva \ea} {\fun {\shiftp {\prodfam \tna}} \tnb}
  \eqdef
  \defunarr
    {\lam {\eva_1}
      {\defunarr
        {\lam {\eva_2} {\subst \ea {\proj \ii \eva} {\eva_i}}^\ii}
        {\fun {\tna_2} \wild}}}
    {\fun {\tna_1} \wild}

  \app {\defunarr \ena {\fun {\shiftp {\prodfam \tna}} \tnb}} {\pairfam \ea}
  \eqdef
  \app
    {\defunarr
      {\app
        {\defunarr \ena {\fun {\tna_1} \wild}}
        {\ea_1}}
      {\fun {\tna_2} \wild}}
    {\ea_2}
  \\
  \defunarr {\lam \eva \ea} {\fun {\shiftp {\shiftn \tpa}} \tna}
  \eqdef
  \defunarr {\lam \eva \ea} {\fun \tpa \tna}

  \app
    {\defunarr \ena {\fun {\shiftp {\shiftn \tpa}} \tna}}
    {(\bo{\dertyinvpp \cxna \emptyset \epa {\shiftn \tpa} \emptyset})}
  \eqdef
  \app {\defunarr \ena {\fun \tpa \tna}} \epa
  \\
  \defunarr \eva \tna
  \eqdef
  \annot \eva {\datatyarr \tna {}}
\end{smallmathparfig}

The definition strongly relies on the inversion of term structure of
focused terms modulo $\parel{\argeq\icc}$. For example, when we define
$\defunarr {\lam \eva \ea} {\fun {\psumfam \tpa} \tna}$, we know that
$\var \eva {\sumfam \tpa}$ and can thus assume modulo
$\parel{\argeq\icc}$ that $\ea$ is of the form
$\fdesumlam \eva \eva {\ea_\idx}$. Similarly in the neutral case, we
know that any neutral
$\dertydown \cxna \ena {\fun {\psumfam \tpa} \tna}$ can only appear in
a term applied to an argument $\dertyup \cxna \epa {\psumfam \tpa}$ --
non-invertible phases are as long as possible, so they cannot stop on
a negative function type -- and we know that such an $\epa$ must be of
the form $\inj \ii \epap$.

The way the focused structure conspires to make this transformation
possible is, in fact, rather miraculous -- is it more than just
a hack? In the $\lam \eva \ea$ case of
$\fun {\shiftp {\prodfam \tna}} \tnb$, we know that a derivation of
the form $\dertydown \cxna \eva {\prodfam \tna}$ can only appear
inside a larger derivation
$\dertydown \cxna {\proj \ii \eva} {\tna_\ii}$, which we can replace
by $\dertydown {\datatyarr \cxna} {\eva_\ii} {\tna_\ii}$. In the
neutral case, a $\annot \ena {\fun {\shiftp {\prodfam \tna}} \wild}$
can only appear applied to an argument
$\dertyup \cxna \epa {\shiftp {\prodfam \tna}}$, but as this is
a shifted negative type we know that $\epa$ is in fact an invertible
form $\dertyinvpp \cxna \emptyset \ea {\prodfam \tna} \emptyset$, and
this must be a pair $\pairfam \ea$, exactly what we needed to define
the transformation. The same phenomenon occurs on the argument
$\dertyup \cxnap \epa {\shiftp {\shiftn \tpa}}$ of the double-shifted
case.

\begin{lemma}[Function elimination preserves inequivalence]
  \label{lem:inequivalence-nofun}
  \begin{smallmathpar}
    (\exists \Cb, \epa, \epap,\ %
      \derdecompto \ea \eap \Cb \epa \epap)
    \\
    \implies
    (\exists \Cb, \epa, \epap,\ %
      \derdecompto {\defun \ea {}} {\defun \eap {}} \Cb \epa \epap)
  \end{smallmathpar}
\end{lemma}
\begin{version}{\Not\cameraversion}
  (The converse implication also holds, but we don't need it.)
\end{version}

\begin{lemma}\label{lem:semantics-nofun}
  \begin{smallmathpar}
    \closedtermsem {\focerase {\defun \ea \ta}}
    ~
    =
    ~
    \closedtermsem {\defun {\focerase \ea} {\polerase \ta}}
    ~
    =
    ~
    \defun {\closedtermsem {\focerase \ea}} {\dataty {\polerase \ta}}
    ~
    \in
    ~
    \closedtypesem {\polerase \ta}
  \end{smallmathpar}
\end{lemma}

\begin{corollary}[Fun-elimination preserves semantic equivalence]
  \label{cor:focused-equivalence-nofun}
  \begin{smallmathpar}
    \closedtermsem \ea = \closedtermsem \eap

    \iff

    \closedtermsem {\defun \ea \ta} = \closedtermsem {\defun \eap \ta}
  \end{smallmathpar}
\end{corollary}

\subsection{Canonicity}
\label{subsec:canonicity}

\begin{definition}
  Let a \emph{closed substitution} $\substa$ be a mapping from
  variables to closed terms of closed types (types without atoms).

  We write $\dersubstty \substa \ca$ when $\ca$ is a closed context
  and, for any $\annot \eva \ta \in \ca$ we have
  $\dernat \emptyset {\Subst \eva \substa} \ta$.

  If $\cca$ is a constraint environment, we say that
  $\dersubstcc \substa \cca$ holds when, for any equation $\ea = \eb$ in
  $\cca$, we have $\Subst \ea \substa \argeq\semeq \Subst \eb \substa$.

  We write $\dersubsttycc \substa \ca \cca$ if
  $\dersubstty \substa \ca$ and $\dersubstcc \substa \cca$ hold.
\end{definition}

\begin{theorem}[Canonicity]
  \label{thm:canonicity}
  Assume $\ca$ is consistent, $\cca$ is valid, and
  \begin{smallmathpar}
   \derneqinv \cap \emptyset {\plug \Cb {\derneqspine \ca \cca \epa \epap
       \tpa}} \tna \txpa
  \end{smallmathpar}
  is a judgment of closed function-less types. Then there
  exists a closed substitution $\dersubsttycc \substa \ca \cca$ such
  that $\Subst \epa \substa \argneq\semeq \Subst \epap \substa$.
\end{theorem}

\begin{corollary}[Canonicity]
  \label{cor:canonicity}
  In the sub-system $\LC{(\sysprodu,\syssumu)}$:
  \begin{smallmathpar}
   \derneqinv \ca \emptyset {\plug \Cb {\derneqspine \cap \ccap \epa \epap
       \tpa}} \tna \txpa
   \Rightarrow
   \plug \Cb \epa \argneq\conteq \plug \Cb \epap
  \end{smallmathpar}
\end{corollary}

\subsection{Results}

\begin{theorem}[Saturated terms are canonical]
  \label{thm:saturated-canonicity}
  In the system with only positive atoms,
  if $\dertysinv \emptyset \cpa {\ea \argneq\icc \eap} \tna \txpa$
  then $\ea \argneq\conteq \eap$.
\end{theorem}

\begin{corollary}[Contextual equivalence implies equality of saturated forms]
  \label{cor:conteq-implies-icc}
  If $\dernat \ca {\ea \bocomma \eap} \ta$ are (non-focused) terms of
  the full simply-typed lambda-calculus $\LC{(\sysfull)}$ with
  $\ea \argeq\conteq \eap$, then for any $\cpa, \tna$ with no positive
  atoms and $\polerase \cpa = \ca, \polerase \tna = \ta$ and any
  saturated terms
  $\dertysinv \emptyset \cpa {\eb \bocomma \ebp} \tna \emptyset$ such
  that $\ea \argeq\betaeta \focerase \eb$ and
  $\eap \argeq\betaeta \focerase \ebp$ we have $\eb \argeq\icc \ebp$.
\end{corollary}

\begin{corollary}\label{cor:conteq-implies-betaeta}
  Contextual and $\betaeta$-equivalence coincide
\end{corollary}

\begin{corollary}\label{cor:decidability}
  Equivalence in the full simply-typed $\lambda$-calculus with sums and
  the empty type is decidable.
\end{corollary}

\begin{corollary}\label{cor:finite-models}
  The full simply-typed $\lambda$-calculus with sums and the empty
  type has the finite model property.
\end{corollary}

\section{Conclusion}

\subsection{Other Related Work}

\paragraph{Background material}

For the reader looking for more historical perspective
\citet*{dougherty2000} gives an interesting, detailed presentation of
the state of the art in separation theorems in absence of sum types,
and of why their approach are difficult to extend to
sums. \citet*{simpson1995} gives an enjoyable, accessible exposition
of what ``completeness'' exactly means in the setting of categorical
semantics; in particular, while we can prove that two inequivalent
terms can be distinguished by a choice of finite sets, there is no
fixed choice of finite sets that could separate all pairs of terms. It
also discusses the notion of ``typical ambiguity'', and the idea that
$\betaeta$-equality is shown to be the maximal consistent relation. At
the time of writing, Simpson's statement of the 15th
\href{http://tlca.di.unito.it/opltlca/}{TLCA open problem} is also one
of the clearest expositions of the relation between these questions.

\paragraph{Focusing}

The main field of application of focusing to programming is the area
of \emph{logic programming}, where operational semantics correspond to
search strategies, which can often be described as specific choices of
polarization in a focused
logic~\citep*{chaudhuri-pfenning-price-2008}.
Focusing has been used to study programming language principles in
\citet*{zeilberger-phd}; when considering non-normal forms (logics with
a strong cut rule) it let us reason finely about evaluation order.

This suggests the general notion of \emph{polarization}, an approach
of programming calculi where compositions (cuts) are
non-associative \citep*{curien2016}, modeling effects and
resources. In the present work we consider only focused normal forms
(except for stating the completeness of focusing), which only captures
the pure, strongly normalizing fragment, and thus corresponds to
a \emph{depolarized} system.

Guillaume Munch-Maccagnoni's work on polarized abstract machine
calculi harps on many ideas close to the present work, in particular
\citet*{MunchMaccagnoni2015}. It is conducted in a syntax that is
inspired by the sequent calculus rather than the $\lambda$-calculus;
this choice gives a beautiful dynamic semantics to polarized
systems. In contrast, the focused $\lambda$-calculus as presented is
a system of normal forms, and defining an internal reduction for it
would be awkward. While there is no direct correspondence between
sequent proofs and natural deduction proofs in general, their normal
forms are in one-to-one mapping, so in our context the choice of
abstract machine or $\lambda$-terms-inspired syntax matters less.

\paragraph{Atom-less systems}

The structure of normal forms has been studied in
\citet*{altenkirch2004} in the special case of the type system
$\LC{(\sysfun, \syssum, 2)}$, with boolean types instead of general
sums, and no atoms. While conceptually easier to follow than the
Grothendieck logical relations of the more general
normalization-by-evaluation work on sums, they remain challenging in
presence of higher-order functions.

In unpublished work, \citet*{ahmad2010} work with the type system
$\LC{(\sysfun,\sysprodu,\syssumu)}$: no atoms, but everything
else. They use focusing as a guiding principle to generalize the
normal forms of $\LC{(\sysfun, \syssum, 2)}$ to everything else -- to
our knowledge this is the first work to use focusing to approach sum
types. The result obtained, namely decidability of observational
equivalence in this system, is not striking on its own: in absence of
atoms, all types are finitely inhabited, so two functions can be
compared by testing them on all their input domain. But it is
conducted in a rigorous, inspiring way that shows the promises of the
focused structure. Our own proof of distinguishability of distinct
normal forms is not as elegant as this development, as it uses the
inelegant shortcut of function type elimination. We are convinced that
there exists a beautiful proof that weaves the distinguishability
structure through higher-order abstractions in their style, but have
yet to find it.

\subsection{Future Work}

\paragraph{Direct distinguishability proof}

The use of the positive atoms simplification and the detour through
types without functions are a bit disappointing. There should be
a proof in which all steps are as general and general as possible:
completeness of focusing in the explicitly polarized system,
completeness of saturated forms, and then canonicity of saturated
forms at all types.

\paragraph{Categorical semantics}

We wonder what is the relation between the saturated structure and the
existing work on categorical semantics of the $\lambda$-calculus with
finite sums.

\paragraph{Direct comparison algorithm}

We prove decidability by reducing equivalence of arbitrary
$\lambda$-terms to equivalence of saturated forms. This algorithm can
be implemented, but we would rather use an algorithm that does not
need to compute full saturated normal forms before returning
a result -- in particular on inequivalent inputs.

It is reasonably easy to formulate an equivalence algorithm on the
focused forms directly, that would perform the saturation ``on the
fly'': at the beginning of each saturation phase, look for all the
neutrals that can be defined in the current context, recursively
compute their equivalence classes, and replace each class by
a distinct free variable -- then continue on the neutral
structure. Proving this algorithm correct, however, turns out to be
challenging.

Finally, it would be interesting to have an algorithm expressed
directly on the non-focused terms, that would perform a focusing
transformation on the fly, as much as each step of equivalence
requires. The difficulty then is that neutral subterms are not as long
as possible (they may be interrupted by commuting conversions), so it
is possible to look for neutrals sub-terms definable in the current
context and miss some of them -- future reasoning stages may un-stuck
sum-eliminations that in turn un-block neutrals that should have been
bound now. For the type system with non-empty sums, control operators
have been used~\citep*{balat2004} to solve that issue; can the
technique be extended?

\acks
\begin{version}{\Anonymous}
  (Redacted for anonymous review.)
  \\
  .\\
  .\\
  .\\
  .\\
  .\\
\end{version}
\begin{version}{\Not\Anonymous}
  Amal Ahmed made this work possible by giving the time, motivation
  and advice that let it foster. Andrew Pitts gave warm encouragements
  to look at the specific question of the empty type, and Alex Simpson
  enthusiastically sat through an intense whiteboard session at
  Schloss Dagstuhl and provided the final motivation to get it
  done. Max New, Didier Rémy and the anonymous referees provided very
  useful feedback on the article.

  This research was supported in part by the National Science
  Foundation (grant CCF-1422133).
\end{version}

\newpage
% We recommend abbrvnat bibliography style.
% \bibliographystyle{abbrvnat}
\bibliographystyle{plainnat}

% The bibliography should be embedded for final submission.
\bibliography{finite-sums}

\newpage

\begin{version}{\Not\cameraversion}
\appendix

\section{Full definitions}
\label{ann:definitions}

\begin{definition*}[\nameonlyref{def:term-semantics}]
  \label{def:annex-term-semantics}

  We define the following naive semantics for term formers:
  \begin{smallmathpar}
  \begin{array}{l}
    \mathtt{var}_\eva
     :
     \typingsem {\ca, \var \eva \ta} \ta \ma
    \\ \mathtt{var}_\eva(\vca)
     \eqdef
     \vca(\eva)

    \\\\ \mathtt{pair}
     :
     \typingsem \ca {\ta_1} \ma
      \times
      \typingsem \ca {\ta_2} \ma
      \to
      \typingsem \ca {\prodfam \ta} \ma
    \\ \mathtt{pair}(f_1, f_2)(\vca)
     \eqdef
     (f_1(\vca), f_2(\vca))

    \\\\ \mathtt{proj}_i
     :
     \typingsem \ca {\prodfam \ta} \ma \to \typingsem \ca {\ta_i} \ma
    \\ \mathtt{proj}_i (f) (\vca)
     \eqdef
     \va_i
    \\   \text{where~} f(\vca) = (\va_1, \va_2)

    \\\\ \mathtt{lam}
     :
     \typingsem {\ca, \var \eva \ta} \tb \ma
      \to
      \typingsem \ca {\fun \ta \tb} \ma
    \\ \mathtt{lam} (f) (\vca)
     \eqdef
     (\va \in \typesem \ta \ma) \mapsto f(\vca, \eva \mapsto \va)

    \\\\ \mathtt{app}
     :
     \typingsem \ca {\fun \ta \tb} \ma
      \times
      \typingsem \ca \ta \ma
      \to
      \typingsem \ca \tb \ma
    \\ \mathtt{app} (f,g) (\vca)
     \eqdef
     f(\vca)(g(\vca))

    \\\\ \mathtt{unit}
     :
     \typingsem \ca \tunit \ma
    \\ \mathtt{unit}(\vca)
     \eqdef
     \star

    \\\\ \mathtt{inj}_i
     :
     \typingsem \ca {\ta_i} \ma
      \to
      \typingsem \ca {\sumfam \ta} \ma
    \\ \mathtt{inj}_i(f)(\vca)
     \eqdef
     (i, f(\vca))

    \\\\ \mathtt{match}
     :
     \typingsem \ca {\sumfam \ta} \ma
      \times
      \typingsem {\ca, \var \eva {\ta_1}} \tb \ma
      \times
      \typingsem {\ca, \var \eva {\ta_2}} \tb \ma
    \\ \qquad\qquad \to \typingsem \ca \tb \ma
    \\ \mathtt{match}(f, g_1, g_2)(\vca)
     \eqdef
     g_i(\vca, \eva \mapsto \va)
    \\ \qquad \text{where~} f(\vca) = (i, \va)

    \\\\ \mathtt{absurd}
     :
     \typingsem \ca \emptyset \ma
      \to
      \typingsem \ca \ta \ma
    \\ \mathtt{absurd}
     \eqdef
     \emptyset
  \end{array}
  \end{smallmathpar}

  By composing together the semantics of the term formers in the
  obvious way, we obtain semantics for terms $\ea$ and one-hole
  contexts $\Ca$:
  \begin{smallmathpar}
    \termsem {\dernat \ca \ea \ta} \ma
    \in
    \typingsem \ca \ta \ma

    \contsem {\dernat \ca {\plug \Ca {\dernat \cap \ehole \tap}} \tap} \ma
    \in
    \typingsem {\ca, \cap} \tap \ma
    \to
    \typingsem \ca \ta \ma
  \end{smallmathpar}
  For example we have
  $\termsem {\pairfam \ea} \ma
  =
  \mathtt{pair}(\termsem {\ea_1} \ma, \termsem {\ea_2} \ma)$
  and
  $\termsem {\pair \ea \ehole} \ma (f)
  =
  \mathtt{pair}(\termsem {\ea_1} \ma, f)$.
  In particular, $\contsem \ehole \ma$ is the identity function.
\end{definition*}

\begin{definition*}[\nameonlyref{subsec:funless}]
  \label{def:annex-fun-less-types}
  ~\\

  (The definitions of $\dataty \ta$ and $\datatyarr \ta$ were given in
  \fullref{fig:dataty-nofun}.)

  \begin{smallmathpar}
  \begin{array}{l}
    \defun \wild \ta \quad:\quad \ta \to \dataty\ta
    \\ \defun {\pairfam \ea} {}
    \eqdef \pairlam {\defun {\ea_\idx} {}}
    \\ \defun {\inj i \ea} {}
    \eqdef \inj i {\defun \ea {}}
    \\ \defun \eunit {}
    \eqdef \eunit
    \\ \defun {\ea} {\fun \ta \tb}
    \eqdef
    \defunarr
      {\lam \eva {\defun {\app \ea {\refun \eva {\ta}}} {\tb}}}
      {\fun \ta \tb}
  \end{array}
  ~
  \begin{array}{l}
    \refun \wild \ta \quad:\quad \dataty\ta \to \ta
    \\ \refun {\pairfam \ea} {}
    \eqdef \pairlam {\refun {\ea_\idx} {}} {}
    \\ \refun {\inj i \ea} {}
    \eqdef \inj i {\refun \ea {}}
    \\ \refun \eunit {}
    \eqdef \eunit
    \\ \refun \ea {\fun {\dataty \ta} {\dataty \tb}}
    \eqdef \refunarr {\lam \eva {\refun {\app \ea {\defun \eva {\ta}}} {\tb}}} {}
  \end{array}
  \end{smallmathpar}

  \begin{smallmathpar}
  \begin{array}{l}
    \defun \wild \ta \quad:\quad \closedtypesem\ta \to \closedtypesem{\dataty \ta} \\

    \defun {(\va_1, \va_2)} {}
    \eqdef (\defun {\va_1} {}, \defun {\va_2} {})
    \\ \defun {(i, \va)} {}
    \eqdef (i, \defun \va {})
    \\ \defun \star {}
    \eqdef \star
    \\ \defun \va {\fun \ta \tb}
    \\ \qquad\eqdef
      \defunarr {\vb \mapsto \defun{\va(\refun \vb {\ta})} {\tb}}
                {\fun \ta \tb}
  \end{array}
  ~
  \begin{array}{l}
    \refun \wild \ta \quad:\quad \closedtypesem{\dataty \ta} \to \closedtypesem\ta \\
    \refun {(\va_1, \va_2)} {}
    \eqdef (\refun {\va_1} {}, \refun {\va_2} {})
    \\ \refun {(i, \va)} {}
    \eqdef (i, \refun \va {})
    \\ \refun \star {}
    \eqdef \star
    \\ \refun {\va} {\fun \ta \tb}
    \\ \qquad\eqdef
      \refunarr {\vb \mapsto \refun{\va(\defun \vb {\ta})} {\tb}}
                {\fun \ta \tb}
  \end{array}
  \end{smallmathpar}

  \begin{smallmathpar}
  \begin{array}{l}
    \defunarr \ea {\fun {\prodfam \ta} \tb}
    \eqdef \defunarr
      {\lam {\eva_1} {\defunarr {\lam {\eva_2} {\app \ea {\pairfam \eva}}} {\fun {\ta_2} \tb}}}
      {\fun {\ta_1} {\datatyarr {\fun {\ta_2} \tb}}}

    \\ \defunarr \ea {\fun \tunit \tb}
    \eqdef \app \ea \eunit

    \\ \defunarr \ea {\fun {\sumfam \ta} \tb}
    \eqdef \pairlam{\defunarr {\lam \eva {\app \ea {(\inj \idx \ea)}}} {\fun {\ta_\idx} \tb}}

    \\ \defunarr \ea {\fun \tempty \tb}
    \eqdef \eunit

    \\[1em] \refunarr \ea {\fun {\prodfam \ta} \tb}
    \eqdef \lam \eva
      {\app
        {\refunarr
          {\app
            {\refunarr
               \ea
               {\fun {\ta_1} {\datatyarr {\fun {\ta_2} \tb}}}}
            {(\proj 1 \eva)}}
          {\fun {\ta_2} \tb}}
        {(\proj 2 \eva)}}
    \\ \refunarr \ea {\fun \tunit \tb}
    \eqdef \lam \eva \ea
    \\ \refunarr \ea {\fun {\sumfam \ta} \tb}
    \eqdef \lam \eva
      {\desumlam \eva {\eva_\idx}
        {\app {\refunarr {\proj \idx \ea} {\fun {\ta_\idx} \tb}} {\eva_\idx}}}
    \\ \refunarr \ea {\fun \tempty \tb}
    \eqdef \lam \eva {\eabsurd \eva}
  \end{array}
  \end{smallmathpar}

  \begin{smallmathpar}
  \begin{array}{l}
    \defunarr \va {\fun {\prodfam \ta} \tb}
    \eqdef \defunarr
      {\semlam {\vb_1} {\defunarr {\semlam {\vb_2}
        {\semapp \va {\sempairfam \vb}}} {\fun {\ta_2} \tb}}}
      {\fun {\ta_1} {\datatyarr {\fun {\ta_2} \tb}}}

    \\ \defunarr \va {\fun \tunit \tb}
    \eqdef \semapp \va \star

    \\ \defunarr \va {\fun {\sumfam \ta} \tb}
    \eqdef \sempairlam{\defunarr {\semlam \vb
      {\semapp \va {(\seminj \idx \va)}}}
      {\fun {\ta_\idx} \tb}}

    \\ \defunarr \va {\fun \tempty \tb}
    \eqdef \star

    \\[1em] \refunarr \va {\fun {\prodfam \ta} \tb}
    \eqdef \semlam {\sempairfam \vb}
      {\semapp
        {\refunarr
          {\semapp
            {\refunarr
               \va
               {\fun {\ta_1} {\datatyarr {\fun {\ta_2} \tb}}}}
            {\vb_1}}
          {\fun {\ta_2} \tb}}
        {\vb_2}}
    \\ \refunarr \va {\fun \tunit \tb}
    \eqdef \semlam \star \va
    \\ \refunarr {\sempairfam \va} {\fun {\sumfam \ta} \tb}
    \eqdef \semlam {\seminj i \vb}
      {\semapp {\refunarr {\va_i} {\fun {\ta_i} \tb}} \vb}
    \\ \refunarr \va {\fun \tempty \tb}
    \eqdef \emptyset
  \end{array}
  \end{smallmathpar}
\end{definition*}

\begin{definition*}[{\hyperref[def:polerase]{Depolarization}}]
  \label{def:annex-polerase}
  \begin{smallmathpar}
  \begin{array}{l@{\quad\eqdef\quad}l}
    \polerase \xpa & \xa \\
    \polerase {\sum \tpa \tpb} & \sum {\polerase \tpa} {\polerase \tpb} \\
    \polerase \tempty & \tempty \\
    \polerase {\shiftp \tna} & \polerase \tna \\
  \end{array}
  ~
  \begin{array}{l@{\quad\eqdef\quad}l}
    \polerase \xnb & \xb \\
    \polerase {\fun \tpa \tna} & \fun {\polerase \tpa} {\polerase \tna} \\
    \polerase {\prod \tna \tnb} & \prod {\polerase \tna} {\polerase \tnb} \\
    \polerase \tunit & \tunit \\
    \polerase {\shiftn \tpa} & \polerase \tpa \\
  \end{array}
  \end{smallmathpar}
\end{definition*}

\section{Proof outlines}
\label{ann:proofs}

\begin{proof}[\fullref{thm:betaeta-implies-semeq}]
  This is proved by direct induction on the evidence that
  $\ea \argeq\betaeta \eb$; semantic equivalence is a congruence, and
  $\beta$-reduction and $\eta$-expansions are identities in the semantic
  domain.
\end{proof}

\begin{proof}[\fullref{thm:semeq-implies-conteq}]
  For any model $\ma$ and context
  $\dernat \emptyset {\plug \Ca \ehole} {\sum \tunit \tunit}$, the
  closed term $\plug \Ca {\modapp \ma \ea}$ must have a normal form
  $\inj \ii \eunit$ by \fullref{lem:lc-inversion}, and
  $\plug \Ca {\modapp \ma \eb}$ a normal form $\inj \jj \eunit$. The
  semantic interpretation of $\plug \Ca \ea$ (technically,
  $\closedtypesem \Ca (\termsem \ea \ma)$ for the natural semantics of
  contexts with a hole) and of $\plug \Ca \eb$ are equal by our
  assumption $\ea \argeq\semeq \eb$ and the fact that semantic
  equality is a congruence. They are also are respectively equal to
  $(\ii, \star)$ and $(\jj, \star)$ by
  \fullref{thm:betaeta-implies-semeq}, so we must have $\ii = \jj$.
  $\plug \Ca \ehole$ is not separating.
\end{proof}

\begin{proof}[\fullref{thm:reification}]
  We define
  $\closedreify \va \eqdef \refun {\closedreifyp {\defun \va \ta}}
  \ta$ using an (obvious) auxiliary function $\closedreifyp \va$ on
  values at fun-less types, on which the inversion property is
  immediate; for $\closedreify \wild$ it is then proved using the fact
  that $\closedtermsem \wild$ and $\defun \wild \ta$, $\refun \wild \ta$
  commute.
\end{proof}

\begin{proof}[\fullrefnoname{lem:reification-inverse}]
  By \fullref{thm:betaeta-implies-semeq}, we can without loss of
  generality consider $\beta$-normal forms only. For the fun-less
  $\closedreifyp \wild$, direct induction proves
  $\closedreifyp {\closedtermsem \ea} = \ea$ on normal-forms $\ea$ --
  note that the fact that our types have no functions is key to
  preserve the inductive invariant that our terms are closed,
  necessary to use the inversion lemma to reason on the shape of
  $\ea$. The result for $\closedreify \wild$ holds by commutation of
  $\closedtermsem \wild$ and $\defun \wild \ta$, $\refun \wild \ta$.
\end{proof}

\begin{proof}[\fullref{thm:conteq-implies-semeq}]
  We prove the contrapositive: if
  $\termsem \ea \ma \neq \termsem \eap \ma$ in some model $\ma$, we
  build a discriminating context $\plug \Ca \ehole$ in $\ma$. Our
  induction on a type $\ta$ has a slightly stronger hypothesis. We
  assume $\dernat {\modapp \ma \ca} {\ea \bocomma \eap} {\modapp \ta \ma}$
  and $\closedtermsem \ea \neq \closedtermsem \eap$: the semantically
  distinct terms are well-typed at the closed typing
  $\dernat {\modapp \ma \ca} \wild {\modapp \ca \ta}$, but need not be
  well-typed at $\dernat \ca \wild \ta$.

  In the function case $\fun \ta \tb$ we have an input
  $\vb \in \typesem \ta \ma$ where the two semantics differ, we build
  the discriminating context
  $\plug \Ca {\papp \ehole {(\reify \vb \ta)}}$,
  where $\Ca$ distinguishes at $\tb$ by induction hypothesis.

  In the sum case, $\termsem \ea \ma$ is of the form $(\ii, \va)$ and
  $\termsem \eap \ma$ of the form $(\jj, \vap)$ and we have either
  $\ii \neq \jj$ or $\va \neq \vap$. In the second case, suppose for
  example $\ii = \jj = 1$, we use a context
  $\pddesum \ehole {\eva_1} {\plug \Ca {\eva_1}} {\eva_2} {\inj \kk
    \eunit}$, using the fact that the return type $\sum \tunit \tunit$
  is inhabited to build a dummy value for the unused branch
  ($\kk \in \{1,2\}$ is arbitrary).
\end{proof}

\begin{proof}[\fullref{lem:strong-cut}]
  By simultaneous induction on $\epa$ and $\ea$. In the variable case,
  it's just a variable/variable substitution.
\end{proof}

\begin{proof}[\fullref{thm:saturation-complete}]
  The proof of \citet*{scherer2015} can be reused almost as-is. It
  defines a rewriting relation from focused terms to saturated terms,
  that proceeds by moving \texttt{let}-bindings around: at any saturation
  point in some context $\cxna; \cxnap$, it looks for all the negative
  neutral terms in $\ea$ that would be well-typed in $\cxna, \cxnap$,
  and saturates over exactly those neutral terms.

  However, just selecting the neutral sub-terms of $\ea$ does not make
  a valid selection function. We also have to add enough neutrals at
  each saturation step to respect the validity criterion; by the
  subformula property, a finite number of neutrals always suffices,
  so there exists a valid selection function with the neutral
  sub-terms of $\ea$.

  Finally, any selection function above this one also satisfies
  completeness. Introducing more neutrals either does not change the
  term shape (after the next invertible phase) or exposes a neutral
  proof of inconsistency $\annot \ena \tempty$; in the latter case,
  more terms become $\betaeta$-equivalent, which helps us proving our
  goal $\focerase \ea \argeq\betaeta \focerase \eap$.
\end{proof}

\begin{proof}[\fullref{thm:saturation}]
  By induction on the derivation. In the base case, the old context is
  empty so the result is trivial -- there is no $\tpa$ such that
  $\dernatdown \emptyset {\shiftn \tpa}$. The induction case considers
  the application of the \Rule{sat} rule to deduce
  $\dertysat \cxna \cxnap {\letin {\bar \eva} {\bar \ena} \ea}
  \txpb$. We have to prove that any sub-derivation of $\ea$ of the
  form $\dertysat {\cxna, \cxnap} {\cxnapp_\kk} {\efa_\kk} \txpb_\kk$
  satisfies the expected property for any $\tpa$.

  If $\dernatdown \cxna \tpa$, this is immediate by induction
  hypothesis. If we have $\dernatdown {\cxna, \cxnap} \tpa$ but not
  $\dernatdown \cxna \tpa$, then $\tpa$ is only provable by a ``new''
  neutral $\ena$ using variables from $\cxnap$, so it is passed to the
  selection function. By validity of the selection function
  (\fullref{fig:selection-function}) we know that some neutral
  $\var \enap \tpa$ must then selected, so $\tpa$ occurs in the
  positive context of the invertible term $\ea$. By
  \fullref{lem:strong-decomposition} we thus have
  $\dernatstrongup {\cxnapp_\kk} \tpa$ as required.
\end{proof}

\begin{proof}[\fullref{lem:saturated-consistency}]
  By (logical) completeness of focusing it suffices to prove that
  there is no \emph{focused} proof of
  $\dertyinvpp \cxna \emptyset \ea \emptyset \tempty$.

  By inversion on the possible derivations, we see that the invertible
  phase must be empty, and \Rule{foclc-let-pos} is the only possible
  choice of focusing: $\ea$ is of the form $\letin \eva \ena \eb$ with
  \begin{smallmathpar}
    \infer
    {\dertydown \cxna \ena {\shiftn \tpb}\\
     \dertyinvpp \cxna \tpb \eb \emptyset \tempty}
    {\dertyfoc \cxna {\letin \eva \ena \eb} \tempty}
  \end{smallmathpar}
  $\cxna$ is saturated so $\dernatdown \cxna {\shiftn \tpb}$ implies
  $\dernatstrongup \cxna \tpb$. But then, by \fullref{lem:strong-cut},
  we know that $\eb$ has a subterm $\dertyfoc \cxna \efa \tempty$.

  We have proved that any term in the judgment
  $\dernatfoc \cxna \tempty$ has a strictly smaller subterm itself in
  the judgment $\dernatfoc \cxna \tempty$. There is no such proof term.
\end{proof}

\begin{proof}[\fullref{thm:inconsistent-canonicity}]
  The inversible and saturation phase are purely type-directed, so
  $\efa$ and $\efap$ necessarily use the same term-formers during
  those phases. They could only differ on a neutral term, but they
  cannot contain a derivation of the form
  $\dertysdown {\cxna, \cxnap} \ena \xna$ or
  $\dertysup {\cxna, \cxnap} \epa \tpa$ as, by
  \fullref{cor:saturation}, $\cxna, \cxnap$ would be saturated,
  which is incompatible with our assumption $\dernatded \cxna \tempty$
  and \fullref{lem:saturated-consistency}.
\end{proof}

\begin{proof}[\fullref{lem:saturated-inequivalence}]
  By direct induction on $\ea \argneq\icc \eap$. In particular, the
  two side-conditions we mentioned are direct consequences of the
  saturated term structure.

  The fact that the $\eva = \ena$ bindings are never redundant is
  a direct consequence of the ``new''-ness condition on saturated
  neutrals: $\ena$ cannot be added in $\cca$ in a future phase as it
  would not be new anymore.

  The fact that the context is consistent at the source of inequality
  is a consequence of \fullrefnoname{cor:saturation} and
  \fullref{lem:saturated-consistency}.
\end{proof}

\begin{proof}[\fullref{thm:inequivalence-in-the-model}]
  We prove this by simultaneous induction on the inequivalence evidence
  for $\eanz \argneq\icc \eanzp$,
  $\derneqinv \cxna \cpa
      {\plug {\Cb_0} {\derneqspine \cxnap \cca {\epanum{0}} {\epapnum{0}} \tpa}}
      \tna \txpb
  $
  and the sub-terms of the two terms $\eano, \eanop$, building
  inequivalence evidence for $\eano \argneq\icc \eanop$. There are
  three interesting cases: when the current term-former introduces
  a variable $\annot \eva \xpa$ in context (or at the root for the
  variables of $\cxna, \cpa$), the saturation step (where
  non-redundancy is checked), and when the source of inequality
  $\derneqspine \cxnap \cca \epa \epap \tpa$ is finally encountered.

  When we enter a binding $\var \eva \xpa$, the corresponding subterms
  of $\eano, \eanop$ start with a large pattern-matching on the
  possible values of $\var \eva {\modapp \nmodlight \xpa}$. We decide
  to take the branch corresponding to the case $\infin \eva$,
  maintaining the inductive invariant that the sub-terms of
  $\eano, \eanop$ are equal to the sub-terms of $\eanz, \eanzp$ under
  the family of substitutions
  $\famx {\subst {} \eva {\infin \eva}} \eva$, for all variables
  $\eva$ of atomic type introduced so far.

  In a saturation step, we must check that the newly added bindings
  are not already present in the constrained environment
  $\ccquotient \cxna \cca$. The neutrals introduced in
  $\eano \argneq\icc \eanop$ are of the form
  ${\subst \ena \eva {\infin \eva}}^\eva$, and this substitution
  preserves $\alpha$-inequivalence, as it replaces distinct variables
  by distinct values.

  When we reach the source of inequality
  $\derneqspine \cxnap \cca \epa \epap \tpa$, the derivation of spine
  inequivalence $\derneqspineup \cxnap \epa \epap \tpa$ either ends up
  on a $\derneqspineup \cxnap {\inj \ii \epa} {\inj \jj \epap} \tpa$
  with $\ii \neq \jj$, in which case the corresponding neutral
  subterms of $\eano, \eanop$ are also of this form, or
  $\derneqspineup \cxnap \eva \evb \xpa$. In this case the subterms of
  $\eano, \eanop$ are $\infin \eva, \infin \evb$, by our inductive
  invariant: they are also distinct positive neutrals, and form
  a source of inequality.

  At the source of inequality, we also have to prove that
  $\ndernatded {\modapp {\nmod \eanz \eanzp} \cap} \tempty$ from the
  hypothesis $\ndernatded \cap \tempty$. Note that this could be wrong
  if we applied some other model $\ma$; for example, the context
  $(\var \eva \xpa, \var \evb {\fun \xpb \tempty})$ is consistent, but
  applying the model $\{ \xpa \mapsto \tunit, \xpb \mapsto \tunit \}$
  results in an inconsistent context.

  Fortunately, the definition of the neutral model
  \fullref{def:neutral-model} is such that the positive atoms that
  become inhabited are exactly those that have bound variables in the
  context at the source of inequality
  $\dernatded \cxnap {\epa \argneq\icc \epap} \tpa$. And, by saturated
  consistency (\fullref{cor:saturation}), we know that those positive
  atoms are exactly those that are \emph{provable} in $\cxnap$.

  We can thus can prove
  $\ndernatded {\modapp \nmodlight \cxnap} \tempty$ by contradiction
  as follows. Suppose we have a proof of
  $\dernatded {\modapp \nmodlight \cxnap} \tempty$. Any subproof of
  the form
  $\dernatded {\modapp \nmodlight \cxnap} {\modapp \nmodlight \xpa}$,
  for one of the $\xpa$ inhabited in $\nmodlight$, can be replaced by
  a subproof of the form
  $\modapp \nmodlight {\dernatded \cxnap \xpa}$, as we know that all
  such $\xpa$ are derivable in $\cxnap$. The transformed proof does
  not rely on the definition of the $\xpa$, so it is also a valid proof
  of $\dernatded \cxnap \tempty$, which is impossible.
\end{proof}

\begin{proof}[\fullref{cor:focused-equivalence-nofun}]
  This is immediate from \fullrefnoname{lem:semantics-nofun} and the fact
  that the function-elimination transformation is a bijection on
  semantics value -- \fullref{fig:dataty-nofun}.
\end{proof}

\begin{proof}[\fullref{thm:canonicity}]
  We remark that $\Subst \epa \substa \argneq\semeq \Subst \epap \substa$
  is a trivial property if $\tpa$ is a closed type (no atoms $\xpa$):
  by inspecting the $\derneqspineup \ca \epa \epap \tpa$ judgment (or
  the strong positive neutral structure), we know that $\epa$ and
  $\epap$ differ on a constructor, and this property is preserved by
  any substitution.

  From $\ndernatded \ca \emptyset$ we know that for any $\ta \in \ca$,
  $\ta$ is inhabited by some closed values --
  \fullref{cor:inhabited-or-inconsistent}. We can always pick an
  arbitrary closed value of type $\ta$.

  We build $\substa$ by peeling off the variables of the environment
  $\ca$, from the most recently introduced to the first introduced. At
  any point during our induction, we have processed a fragment $\cbp$
  of the context, and $\cb$ remains unprocessed, with
  $\ca = (\cb, \cbp)$, and we have a substitution $\substa$ that binds
  the variables of $\cbp$ to closed terms, such that
  $\dersubstcc \substa \cca$ contains no inconsistent
  equation. Initially we have $\cb \eqdef \ca$,
  $\cbp \eqdef \emptyset$, $\substa \eqdef \emptyset$.

  When $\ca$ is of the form $\cap, \var \eva {\sumfam \ta}$, we look
  for equations of the form $\inj \ii {\eva_\ii} = \eva$ in $\cca$; by
  validity of $\cca$ (\fullrefnoname{def:valid-constraints}), there is at
  most one. If there is none, we choose an arbitrary value for $\eva$
  in $\substa$; if there is one, we choose
  $\Subst {\eva_\ii} {\substa}$; this is a closed term as $\eva_\ii$
  must be defined after $\eva$ in the context.

  When $\ca$ is of the form $\cap, \var \eva \tna$, we look for
  equations the family of equations
  $\fami{\evb_\ii = \plug {\ena_\ii} \eva}$ -- remember that we write
  $\plug \ena \eva$ when the neutral $\ena$ has head $\eva$. Let us
  write $\plug \cca \eva$ for this set of equations about $\eva$.

  We define the value for $\eva$ in $\substa$ as
  $\Val \eva \tna {\fami {\evb_\ii = \ena_\ii}}$ where
  $\Val \ena \tna \cca$, defined below, is such that
  $\dersubstcc {{\Val \ena \tna \cca} / \ena} \cca$ always holds, if
  all equations in $\cca$ are of the form $\evb = \plug \enb \ena$.
  \begin{smallmathpar}
    \Val \ena \tunit \cca \eqdef \eunit

    \Val \ena {\prodfam \tna} \cca
    \eqdef~
    \vpairlam
      {\!\!\Val
        {\proj \idx \ena}
        {\tna_\idx}
        {\{ (\eva = \plug \enb {\proj \idx \ena}) \in \cca \}}\!\!}

    \Val \ena {\shiftn \tpa} {\{ \evb = \ena \}}
    \eqdef \evb

    \Val \ena {\shiftn \tpa} \emptyset \eqdef \text{arbitrary}
  \end{smallmathpar}
  From the validity of $\cca$, we know that $\plug \cca \eva$ contains
  at most one equation of the form $\annot \evb \tpa = \ena$ for each
  possible neutral $\dertydown {} \ena {\shiftn \tpa}$, which justifies
  the only two cases in the definition of $\Val \ena {\shiftn \tpa} \cca$ above.
\end{proof}

\begin{proof}[\fullref{cor:canonicity}]
  We use a stronger induction hypothesis on the inequivalence
  judgment: if
  \begin{smallmathpar}
  \derneqinv \ca \cca {\plug \Cb {\derneqspine \cap \ccap \epa \epap
     \tpa}} \tna \txpa
  \end{smallmathpar}
  then there is a substitution $\dersubsttycc \substa \ca \cca$ and
  a context $\plug \Ca \ehole$ such as
  \begin{smallmathpar}
    \dernat \emptyset
    {\plug \Ca {\Subst \ea \substa}
     \argneq\beta
     \plug \Ca {\Subst \eap \substa}} \tbool
  \end{smallmathpar}%
  where $\ea \eqdef \plug \Cb \epa$ and $\eap \eqdef \plug \Cb \epap$.

  In the base case $\Cb \eqdef \ehole$, we get a substitution
  $\dersubsttycc \substa \cap \ccap$ from \fullref{thm:canonicity};
  the theorem also gives us
  $\Subst \epa \substa \argneq\semeq \Subst \epap \substa$, so there
  is a distinguishing context $\Ca$ as required as contextual and
  semantic equivalence coincide.

  In the sum-elimination case
  \begin{smallmathpar}
  \infer
  {\derneqinv
    {\ca, \var \eva {\sumfam \tpa}, \var {\eva_\ii} \tpa_\ii}
    {\cca, \inj \ii {\eva_\ii} = \eva}
    {\plug \Cb {\derneqspine \cap \ccap \epa \epap \tpap}}
    \tna \txpb}
  {\derneqinv {\ca, \var \eva {\sumfam \tpa}} \cca
    {\plug
      {\left({\bcasenoinj \eva
        {\inj \ii {\eva_\ii}} \ehole
        {\inj \jj {\eva_{\jj \neq \ii}}} \ea}\!\!\!\!\right)}
      {\derneqspine \cap \ccap \epa \epap \tpap}
    }
    \tna \txpb}
  \end{smallmathpar}
  the substitution $\substa$ verifies, by inductive hypothesis, that
  $\inj \ii {\Subst {\eva_\ii} \substa} = \Subst \eva \substa$, so in
  particular the sum-elimination term, once applied $\substa$, will
  reduce to $\Subst {(\plug \Cb \epa} \substa$ and
  $\Subst {(\plug \Cb \epap)} \substa$ respectively, as desired: this
  means that the distinguishing context $\Cb$ for the premise is also
  a distinguishing context for the sum-elimination terms.

  In the (left) pair case
  \begin{smallmathpar}
    \infer
    {\derneqinv \ca \cca
      {\plug \Cb {\derneqspine \cap \ccap \epa \epap \tpa}} {\tna_1} \emptyset}
    {\derneqinv \ca \cca
      {\plug {\pair \Cb \eb}
        {\derneqspine \cap \ccap \epa \epap \tpa}} {\prodfam \tna} \emptyset}
  \end{smallmathpar}
  We can obviously keep the same $\substa$, and build the
  distinguishing context $\plug \Ca {\proj 1 \ehole}$, where $\Ca$ is
  the distinguishing context for the inductive premise.

  From this stronger inductive hypothesis $(\substa, \Ca)$, we can
  build a distinguishing context for
  $\plug \Cb \epa \argneq\conteq \plug \Cb \epap$ as
  \begin{smallmathpar}
    \plug \Ca {\app
      {\plam {(\eva_1 \dots \eva_n \in \ca)} \ehole}
      {(\Subst {\eva_1} \substa) \dots (\Subst {\eva_n} \substa)}}
  \end{smallmathpar}
\end{proof}

\begin{proof}[\fullref{thm:saturated-canonicity}]
  Let $\eanz \eqdef \ea$ and $\eanzp \eqdef \eap$.

  By \fullref{lem:saturated-inequivalence} we have
  $\derneqinv \cpa \emptyset
      {\plug {\Cb_0} {\derneqspine \cxnap \cca {\epanum 0} {\epapnum 0} \tpa}}
      \tna \txpb
  $
  with $\eanz \argeq\icc \plug {\Cb_0} {\epanum 0}$ and
  $\eanzp \argeq\icc \plug {\Cb_0} {\epapnum 0}$, that is,
  $\derdecompto \eanz \eanzp {\Cb_0} {\epanum 0} {\epapnum 0}$.

  We consider the neutral model $\nmod \eanz \eanzp$ defined in
  \fullref{def:neutral-model}. Let $\eano \eqdef \modapp \nmodlight \eanz$ and
  $\eanop \eqdef \modapp \nmodlight \eanzp$.

  By \fullref{thm:inequivalence-in-the-model} we have
  $\derdecompto \eanz \eanzp {\Cb_1} {\epanum 1} {\epapnum 1}$,
  that is a inequivalence of terms at the atom-less judgments
  $\dertyinvpp \emptyset {\modapp \nmodlight \cpa}
    {\eanz \argneq\icc \eanzp} {\modapp \nmodlight \tna} {\modapp \nmodlight \txpb}
  $
  .

  Using the focused function-elimination transformation of
  \fullref{fig:focused-nofun}, let us define
  $\eant \eqdef \defun \eano {}$ and $\eantp \eqdef \defun \eanop {}$.
  By \fullref{lem:inequivalence-nofun} we get
  $\derdecompto \eant \eantp {\Cb_2} {\epanum 2} {\epapnum 2}$,
  that is a saturated inequivalence at closed types without functions.

  From \fullref{cor:canonicity}, we finally get
  $\plug {\Cb_2} {\epanum 2} \argneq\conteq \plug {\Cb_2} {\epapnum 2}
  $,
  that is $\eant \argneq\conteq \eantp$, a contextual inequivalence at
  closed types without functions. As contextual and semantic
  equivalence coincide at closed types, this is equivalent to
  $\closedtermsem \eant \neq \closedtermsem \eantp$,
  that is,
  $\closedtermsem {\defun \eano {}} \neq \closedtermsem {\defun \eanop {}}$.

  By \fullref{cor:focused-equivalence-nofun}, we thus have that
  $\closedtermsem \eano \neq \closedtermsem \eanop$ at closed types,
  that is,
  $\termsem \eanz \nmodlight \neq \termsem \eanzp \nmodlight$, that
  is, $\eanz \argneq{\semeqmod \nmodlight} \eanzp$, that is,
  $\eanz \argneq\semeq \eanz$, or equivalently
  $\eanz \argneq\conteq \eanz$.
\end{proof}

\begin{proof}[\fullref{cor:conteq-implies-icc}]
  By contrapositive, we show that $\eb \argneq\icc \ebp$ implies
  $\ea \argneq\conteq \eap$. If we had $\eb \argneq\icc \ebp$ we would
  have $\focerase \eb \argneq\conteq \focerase \ebp$ by
  \fullref{thm:saturated-canonicity}. As $\betaeta$-equivalence
  implies contextual equivalence, we have
  $\ea \argeq\conteq \focerase \eb$ and
  $\eap \argeq\conteq \focerase \ebp$, so $\ea \argneq\conteq \eap$ by
  transitivity.
\end{proof}

\begin{proof}[\fullref{cor:conteq-implies-betaeta}]
  We know from \fullref{sec:full-calculus} that $\betaeta$-equivalence
  implies contextual equivalence. Conversely, we suppose
  $\ea \argeq\conteq \eap$ and prove $\ea \argeq\betaeta \eap$.

  By \fullref{thm:focusing-complete}, there exists focused forms
  $\eb, \ebp$ that erase to $\ea, \eap$ (respectively) modulo
  $\betaeta$; we can even pick them in the subset with positive atoms
  only.

  By \fullref{thm:saturation-complete}, there exists saturated forms
  $\ec, \ecp$ that have the same erasures as $\eb, \ebp$
  (respectively) modulo $\betaeta$, that is they erase to $\ea, \eap$
  modulo $\betaeta$.

  By \fullref{cor:conteq-implies-icc}, we have $\ec \argeq\icc
  \ecp$. As the invertible commuting conversions are sound with
  respect with $\betaeta$-equality, we have
  $\focerase \ec \argeq\betaeta \focerase \ecp$, that is,
  $\ea \argeq\betaeta \eap$.
\end{proof}

\begin{proof}[\fullref{cor:decidability}]
  To test if two terms are equivalent, we invoke the focusing
  completeness result then the saturated completeness result, and
  check if the normal form coincide modulo $\parel{\argeq\icc}$. Those
  two transformations are computable and $\parel{\argeq\icc}$ is
  decidable, so equivalence is decidable.
\end{proof}

\begin{proof}[\fullref{cor:finite-models}]
  Our notion of semantic equivalence coincides with $\betaeta$ and
  contextual equivalence, and only distinguishes terms using sets
  isomorphic to closed types, which are all finitely inhabited.
\end{proof}
\end{version}

\end{document}